\documentclass[12pt,a4paper,fleqn]{article}

\usepackage{a4wide}
\usepackage{amsmath,amsthm}
\theoremstyle{plain}
\newtheorem{thm}{Theorem}[section]
\newtheorem{lem}[thm]{Lemma}

\theoremstyle{definition}
\newtheorem{defn}[thm]{Definition}

\newtheorem{nota}[thm]{Notation}
\newtheorem{remk}[thm]{Remark}

\usepackage{graphicx}
\usepackage[pagebackref=false,colorlinks,linkcolor=blue,citecolor=blue]{hyperref}
\usepackage{ifthen}
\usepackage{mathpartir}
\usepackage{natbib}
\usepackage{setspace}
\usepackage{stmaryrd}
\usepackage{url}

\usepackage[utf8]{inputenc}
\usepackage[T1]{fontenc}
\usepackage[tt=false, type1=true]{libertine}
\usepackage[libertine]{newtxmath}
\usepackage[sans]{dsfont}
\usepackage[mathscr]{eucal}

\newcommand{\agdadef}[3]{\textnormal{\href{\agdaroot#1.html\##2}{\texttt{#3}}}}
\newcommand{\agdalink}[1]{\textnormal{\href{\agdaroot#1.html}{\texttt{#1}}}}
\newcommand{\agdalinkalt}[2]{\textnormal{\href{\agdaroot#1.html}{\texttt{#2}}}}
\newcommand{\agdaroot}{https://amp12.github.io/IntensionalSetoids/html/}

\newcommand{\Atom}{\mathbb{A}}
\renewcommand{\c}{\mathtt{c}}
\newcommand{\C}{\mathscr{C}}
\newcommand{\cng}{\mathtt{cng}}
\newcommand{\coe}{\mathtt{coe}}
\newcommand{\coh}{\mathtt{coh}}
\newcommand{\centeredcirc}[1]{\vcenter{\hbox{$#1\circ$}}}
\newcommand{\smallcirc}{
  \mathbin{\mathchoice{
      \centeredcirc\scriptstyle}{
      \centeredcirc\scriptstyle}{
      \centeredcirc\scriptscriptstyle}{
      \centeredcirc\scriptscriptstyle}}}
\newcommand{\comp}{\smallcirc}

\newcommand{\cons}{\mathrel{;}}

\newcommand{\den}[1]{\llbracket#1\rrbracket}
\newcommand{\dencx}[1]{\llbracket#1\rrbracket}
\newcommand{\dentm}[1]{\llbracket#1\rrbracket}
\newcommand{\denty}[1]{\llparenthesis#1\rrparenthesis}
\newcommand{\dencxr}[1]{\mathop{\llbracket#1\rrbracket{\searrow}}}
\newcommand{\densbr}[1]{\mathop{\llbracket#1\rrbracket{\searrow}}}
\newcommand{\dentmr}[1]{\mathop{\llbracket#1\rrbracket{\searrow}}}
\newcommand{\dentyr}[1]{\mathop{\llparenthesis#1\rrparenthesis{\searrow}}}

\newcommand{\E}{\mathscr{E}}
\newcommand{\El}{\mathscr{E}\!\ell}
\newcommand{\Elem}{\mathscr{E}\!\mathit{lem}}
\newcommand{\Emp}{\mathtt{Emp}}
\newcommand{\Empt}{\mathit{Emp}}
\newcommand{\empt}{\mathit{emp}}

\newcommand{\ent}{\vdash}
\newcommand{\eq}{\sim}
\newcommand{\Eq}{\mathtt{Eq}}
\newcommand{\Eqt}{\mathit{Eq}}
\newcommand{\ETU}{\textnormal{\textsf{ETU}}}
\newcommand{\ETUapp}[4]{#1\cdot_{#2,#3}#4}
\newcommand{\ETUEmp}{\mathbf{Emp}}
\newcommand{\ETUECx}{\diamond}
\newcommand{\ETUemp}{\mathbf{emp}}
\newcommand{\ETUEq}{\mathbf{Eq}}
\newcommand{\ETUlam}{\lambda}
\newcommand{\ETUNat}{\mathbf{Nat}}
\newcommand{\ETUnrec}{\mathbf{nrec}}
\newcommand{\ETUPi}{\Pi}
\newcommand{\ETUrefl}{\mathbf{refl}}
\newcommand{\ETUsucc}{\mathbf{succ}}
\newcommand{\ETUU}{\mathbf{U}}
\newcommand{\ETUzero}{\mathbf{zero}}
\newcommand{\Fam}{\mathscr{F}\!\mathit{am}}
\newcommand{\fib}{\mathbin{'}}
\newcommand{\ffib}{\mathbin{''}}
\newcommand{\Fin}{\mathop{\mathtt{Fin}}}
\newcommand{\freshfor}{\mathrel{\#}}
\newcommand{\FsetA}{\mathtt{Fset}\Atom}

\newcommand{\fun}{\shortrightarrow}
\newcommand{\hcng}{\mathtt{hcng}}
\newcommand{\heq}{\approx}
\newcommand{\Hom}{\mathscr{H}\mathit{om}}
\newcommand{\hrfl}{\mathtt{hrfl}}
\newcommand{\hsym}{\mathtt{hsym}}
\newcommand{\htrs}{\mathtt{htrs}}

\newcommand{\IRU}{\textnormal{\textsf{IRU}}}

\newcommand{\Lft}{\mathtt{Lft}}

\newcommand{\morphism}{\rightarrow}
\newcommand{\Nat}{\mathtt{Nat}}

\newcommand{\NNO}{\mathds{N}}
\newcommand{\Ok}{\mathop{\mathtt{Ok}}}
\newcommand{\Pit}{\mathtt{Pi}}
\newcommand{\refl}{\mathtt{refl}}
\newcommand{\reflect}{\mathit{reflect}}
\newcommand{\rfl}{\mathtt{rfl}}
\newcommand{\rflt}{\mathit{refl}}
\newcommand{\Set}{\mathtt{Set}}
\newcommand{\Setd}{\mathtt{Setd}}
\newcommand{\SetOne}{\Set_1}
\newcommand{\Sigt}{\mathtt{Sigma}}

\newcommand{\soundTm}{\mathtt{sound}\den{\mathtt{tm}}}
\newcommand{\soundTy}{\mathtt{sound}\den{\mathtt{ty}}}

\newcommand{\sym}{\mathtt{sym}}
\newcommand{\Tm}[1]{\mathtt{Tm}[#1]}
\newcommand{\TM}{\mathtt{Tm}}
\newcommand{\totCx}{\mathtt{tot}\den{\mathtt{cx}}}
\newcommand{\totTm}{\mathtt{tot}\den{\mathtt{tm}}}
\newcommand{\totTy}{\mathtt{tot}\den{\mathtt{ty}}}
\newcommand{\trs}{\mathtt{trs}}

\renewcommand{\tt}{\mathtt{tt}}
\newcommand{\Ty}[1]{\mathtt{Ty}[#1]}
\newcommand{\TY}{\mathtt{Ty}}
\newcommand{\U}{\mathscr{U}}
\newcommand{\uip}{\mathit{uip}}
\newcommand{\Unit}{\mathtt{Unit}}
\newcommand{\univ}{\mathit{univ}}
\newcommand{\Univ}{\mathtt{Univ}}

\title{Setoids in Intensional Type Theory}

\author{Andrew M. Pitts\\
  University of Cambridge, UK}

\date{}

\begin{document}

\maketitle

\begin{abstract}
  We show that a certain notion of displayed setoid (family of
  setoids) in intensional type theory can be used to give a semantics
  for extensional type theory with universes (\ETU). Safe Agda serves
  as a machine-checkable formalization of intentional type theory
  augmented with a universe closed under inductive-recursive
  definitions (\IRU). The syntax of \ETU{} is defined in \IRU{} in a
  traditional extrinsic form, using a well-scoped locally nameless
  representation of its terms. Giving the semantics of \ETU{} in terms
  of displayed setoids is complicated by the very limited means of
  expression afforded by \IRU. As a corollary we obtain a proof within
  \IRU{} of the consistency of \ETU.
\end{abstract}

\section{Introduction}

This paper constructs a model of extensional type theory with
universes \citep{Martin-LoefP:inttt} that we call \ETU, within
Martin-L{\"o}f's intentional type theory~\citep{Martin-LoefP:intttp}
augmented with a universe closed under inductive-recursive
definitions~\citep{DybjerP:genfsi}. To be more precise, the type
theory in which the construction takes place, which in the paper will
be called \IRU, is that provided by the Agda interactive theorem
prover~\citep{agda} with options \verb"--safe" and
\verb"--without-K".\footnote{The first option switches off
  experimental and possibly inconsistent features; the second option
  rules out the use of uniqueness of identity proofs in dependent
  pattern matching~\citep{CockxJ:patmwk}. To keep within the simplest
  possible type theory, we also avoid by hand the use of Agda's
  coinductive types and its annotations for proof irrelevance, even
  though Agda's safe option allows both.} We give the syntax of \ETU{}
in a traditional extrinsic form: an inductive definition of its
expressions and then inductive definitions of its well-formed
judgements. The use of a well-scoped locally nameless representation
of syntax~\citep{PittsAM:welsln} allows the inductive definition of
judgements to be quite close to informal, text book definitions. The
model construction makes use of a notion of
\emph{displayed}\footnote{We follow~\citet{AhrensB:disc} and say
  ``displayed structures'' rather than ``family of structures''.}
setoids that is simpler than many that occur in the literature (and
which has also been used in recent work by \citet{PujetL:indus}).  The
semantics of \ETU{} in the model are complicated by the very limited
means of expression afforded by \IRU; the use of machine-assisted
proof was essential for getting both workable definitions and correct
proofs of their properties. As a corollary we obtain a machine-checked
proof within safe Agda of the consistency of \ETU.

Reading the above paragraph, the knowledgeable reader may be surprised
that the literature does not already contain such a setoid model. The
author needed it for other results about the meta-theory of type
theories (namely to try to obtain an intensional version of the
constructions of~\citet{CoquandT:canndt}) and was surprised not to find
it. The closest approach is the posthumously published draft paper by
\citet{PalmgrenE:typtsb} (which is also accompanied by an Agda
formalization). In a Foreward to that paper its editors and reviewer
assert that it
\begin{quote}
  ``provides a formalized interpretation of the extensional version of
  Martin-Löf type theory in an intensional version and \ldots such an
  interpretation has foundational significance.''
\end{quote}
I agree with the second part of this assertion, but Palmgren's
Introduction immediately makes clear that the first part needs
qualifying: the setoid constructions he gives in Agda are shown to be
closed under certain operations that might be used to give the
semantic functions on the expressions of extensional type theory and
to prove that its derivable judgements are satisfied by the model; but
the semantic functions and proof of soundness are not actually
given. Since one wants a semantics of the well-formed expressions of
type theory, rather than a semantics of derivations of
well-formedness\footnote{Recall that there can be many different
  proofs of well-formedness of any particular expression in context,
  and that for extensional type theory the existence of such proofs is
  in general undecidable~\citep[section~3.2.2]{HofmannM:extcit}.}, the
common approach, going back to \citet{StreicherT:semtt}, is to define
\emph{partial} semantic functions and then prove that they are totally
defined on expressions for which there is a proof of
well-formedness. As \citet[section~3.5]{HofmannM:synsdt} remarks, the
details ``although essentially straightforward present surprising
technical difficulties'' and are very often left to the reader.  Very
few machine-checked formalizations exist in the literature; one occurs
as part of the proof of an Initiality Conjecture by Brunerie and
\citet{deBoerM:profic} within Agda augmented with postulated quotients
and propositional extensionality. Here we surmount some ``surprising
technical difficulties'' and are able to define a Streicher-inspired
semantics of \ETU{} within the very bare-bones system of \IRU, dealing
with partial functions via their type-valued graphs and so having to
prove single-valuedness as well as totality. Intensionality causes the
difficulties: just as there can be several constructs in constructive
mathematics that are equivalent classically and it can be tricky to
choose the right one for a given purpose, so here, there can be
several intensional constructs that are equivalent extensionally and
experimentation (with machine assistance) is needed to find the right
one. All of which is to say that the semantics of \ETU{} in \IRU{}
that we give in section~\ref{sec:semetu} was hard won.

Our setoid model of \ETU{} differs from that implicit in
\citet{PalmgrenE:typtsb} in two respects, even though both use
type-valued equivalence relations (``proof-relevant'' setoids), rather
than ones valued in some type of (possibly definitionally
proof-irrelevant) propositions. First and most importantly, the notion
of displayed setoid that we use (see Definition~\ref{def:dsetd}) is much
simpler than Definition~3.8 of \citet{PalmgrenE:typtsb}. Secondly, we
deal differently with setoid universes. Palmgren indirects through set
theory, constructing a large setoid of constructive iterative sets in
the sense of \citet{AczelP:typtic-cp} that supports a set-theoretic
semantics of \ETU. Instead we used the more direct approach of
Altenkirch~\citep{AltenkirchT:exteit,AltenkirchT:conusm}, adapted from
the setting of equivalence relations valued in a universe of
definitionally proof-irrelevant propositions to the simpler setting of
type-valued equivalence relations in \IRU. As for Palmgren, our setoid
universes in \IRU{} (a countable hierarchy of them) are given by
inductive-recursive definitions (and an outer recursion over the
natural numbers for the hierarchy). \citet{AltenkirchT:conusm} show
for their version of displayed setoids how to replace the use of a
universe closed under induction-recursion by a higher-level universe
just closed under induction. \citet{PujetL:indus} goes further and
shows how to do that for a similar notion of setoid universe to the
one we use here. Nevertheless, for the sake of conceptual simplicity
(and because dealing with a countable hierarchy of universes
complicates matters), here we do not try to eliminate the use of
induction-recursion.

The motivation for the earliest work on setoids in intensional
dependent type theory by \citet{HofmannM:extcit} and
\citet{AltenkirchT:exteit} is somewhat different from here. There the
emphasis is on constructing within intensional type theory models of
extensionality principles (such as for function types) whose algebraic
properties hold as much as possible up to definitional equality rather
than just propositional equality or setoid equality (the latter being
ultimately all we care about in this paper).  The aim is to derive
from such models new type theories extending intensional type theory
with extensionality principles without disturbing its good
computational properties, such as canonicity and having decidable
judgements; see for example
\citep{AltenkirchT:obsen,AltenkirchT:settts,SterlingJ:cublbs}.

Extensional type theory does not have those good computational
properties, but it does occur naturally as the basis of internal type
theories of many useful mathematical structures, such as toposes; and
it bears a close relationship to intensional type theory extended with
postulated function extensionality and uniqueness of identity
proofs~\citep{WinterhalterT:elirtt,KapulkinK:extcit}.  Even if the
model of \ETU{} in \IRU{} presented here may not have direct
computational applications, we believe that it is of foundational
interest. Although intentional type theory is logically quite weak, it
can be used, and indeed is used in practice, to represent the
\emph{syntax} of various languages from logic and computer science;
but it is not completely clear to what extent \emph{semantic}
constructions for those languages, relying on extensionality, can be
carried out within it. This paper clarifies the situation and as such
it follows in a line of work on relative consistency of extensionality
in intensional systems that began with \citet{TakeuitG:genlc} and
\citet{GandyR:axieI} (or even Turing, according
to~\citet{HylandJME:fort}), continued with \citet{HofmannM:extcit} and
\citet{AltenkirchT:exteit} and now includes the homotopical view of
equality in type theory~\citep{BarrasB:gentgi}.

\subsection*{Outline}

The Agda code described in this paper is available at
\citep{setitt-agda} and we give links like \agdalinkalt{README}{this}
to relevant parts of it (converted to html) throughout this paper.
Section~\ref{sec:itu} gives a brief summary of the limited subset of
Agda type theory that we use, called \IRU; we aim to make the paper as
accessible as possible to a reader familiar with that type theory, but
not familiar with Agda's concrete implementation of
it. Section~\ref{sec:diss} explains the notions of setoid and
displayed setoid that we use. Section~\ref{sec:setu} gives a countable
hierarchy of inductive-recursively defined universes of codes for such
setoids, closed under taking dependent products, extensional equality
types and containing empty and natural number types; at the top of the
hierarchy there is a universe of codes for typing contexts.
Section~\ref{sec:extu} recalls extensional type theory and discusses
the terms and the judgements of the version we use here, called \ETU.
Section~\ref{sec:semetu} uses the setoid universes to give a semantics
in \IRU{} for \ETU{} and to prove its consistency
(Theorem~\ref{thm:consis}). The fully formalized intensional semantics
of extensional type theory with universes is the main
contribution of this paper, but it also highlights the usefulness of
the notion of displayed type-valued setoid in intensional type theory
that we use for the semantics.

\section{Safe Agda (\IRU)}
\label{sec:itu}

The Intentional Type Theory of the title is that provided by the Agda
interactive theorem prover~\citep{agda}. Although the latter has proved
to be a useful test bed for all sorts of ideas in type theory, 
here we use its \verb+--safe+ option (as well as \verb+--without-K+,
see below) to switch off experimental features
and use the following core type theory, that we call \IRU:
  
\paragraph*{Dependent function types} These are written
$(x:A)\fun B\,x$ in Agda rather than $\prod_{x:A} B\,x$ and satisfy
$\eta$-conversion in addition to $\beta$-conversion.

\paragraph*{Dependent record types}
We use the notation $\sum[x \in A ] B\, x$ for dependent product types
($\Sigma$-types) that classify pairs $(a,b)$ where $a : A$ and
$b : B\,a$. The accompanying projection functions are written
\begin{gather*}
  \pi_1 : (z :\textstyle\sum[x \in A ] B\, x) \fun A\\
  \pi_2: (z :\textstyle\sum[x \in A ] B\, x) \fun B(\pi_1\,z)
\end{gather*}
The non-dependent special case is written using the usual notation for
cartesian product, $A \times B$.  These notations are ones we have
chosen (see \agdalink{Prelude.Product}) rather than being fixed by
Agda, because dependent products are instances of user-definable
dependent record types~\citep{PollackR:deptrt}. These provide a single
constructor and named and typed fields for projecting out of the
record type, all declared by the user. Record types satisfy
$\eta$-conversion\footnote{More precisely, inductive record types do,
  by default. We avoid the other kind of recursive
  record type, by never using the \texttt{coinductive}
  directive.}  (surjective pairing in the particular case of dependent
product types).

\paragraph*{Data types} An Agda user can specify various kinds of
inductive data by declaring the names and the types of functions that
construct terms of the data type. Elimination of values of such types
is achieved very conveniently by defining functions of the data whose
arguments are specified using (dependent) patterns. The original form
of dependent pattern-matching by~\citet{CoquandT:patmdt} was too
liberal (it implies uniqueness of identity proofs, or equivalently,
the Axiom~K of \citet{SteicherT:invitt}), but the option
\verb+--without-K+ implements a stricter form~\citep{CockxJ:patmwk}
equivalent to having the usual eliminators and computation rule for
data in Martin-L\"of type theory~\citep{GoguenH:elidpm}. Agda's scheme
for making data type declarations is very (perhaps too) liberal: as
well as parameterized, indexed and simultaneous definitions, it
permits induction-recursion~\citep{DybjerP:genfsi}, where a type is
inductively defined simultaneously with recursively defined functions
on it (possibly valued in a universe type), and we make use of this.

\paragraph*{Universes} Agda contains universe types, written
$\Set\,\ell$ (where $\ell$ ranges over expressions for universe
levels). They are ``Russellian'' universes, in the sense that there is
no notational distinction between a term of type $\Set\,\ell$ and the
type that it denotes. Universes form a hierarchy in which $\Set\,\ell$
is an element of the universe at the next higher level; but they are
not cumulative (terms of one universe are not literally terms of
higher-level universes). Instead, typing constructs potentially
increase levels; for example, if $A:\Set\,\ell$ and $B$ maps each
$x:A$ to a type in universe $\Set\,\ell'$, then the dependent product
$(x:A)\fun B\,x$ lives in the universe whose level is the least upper
bound of $\ell$ and $\ell'$. Universes are similarly closed under
forming record types, inductive data types and inductive-recursive
types. Because the latter are so powerful, we only need the least
universe in the hierarchy, denoted $\Set$, and the next one above it,
which is denoted $\SetOne$ and used here just to classify collections
of types in $\Set$ (such as the record type $\Setd$ of all setoids in
$\Set$). Thus we largely avoid the ``level yoga'' of
\citet{PalmgrenE:typtsb}, needed to maintain predicativity when
considering setoids whose underlying type is in one universe and whose
equivalence relation is valued in another.

\begin{nota}
  Agda permits mixfix operations with ``$\_$'' used to indicate
  argument positions; for example in Definition~\ref{def:setd} below,
  $A\ni\_{\eq}\_$ is a function of two arguments.
  
  A dependent function type $(x:A)\fun B\,x$ can be written
  $\forall x \fun B \,x$ when the type $A$ is inferable from the
  context. Arguments to functions can be declared to be implicit using
  the syntax $\{x:A\}\fun B\,x$; thus if $f: \{x:A\}\fun B\,x$ and
  $a:A$, then the application $f\{a\}$ has type $B\,a$ can can be
  written just as $f$ with $a$ omitted, provided that there is enough
  contextual information for the Agda system to infer it. Iterated
  dependent function types like $(x:A) \fun((y: B\,x) \fun C\,x\,y)$
  can be written using telescopes, $(x:A)(y : B\,x) \fun C\,x\, y$. A
  non-dependent function type is just written as $A\fun B$.

  We write $\emptyset$ for the empty type in $\Set$, a datatype with
  no constructors; see \agdalink{Prelude.Empty}. The negation of a
  type, $\neg A$, is an abbreviation for the function type
  $A \fun\emptyset$.

  We write $\top$ for a one-element type in $\Set$, a record type with
  no fields whose unique constructor is written $\tt$; see
  \agdalink{Prelude.Unit}.

  We write $\NNO$ for the type of natural numbers in $\Set$, a
  datatype with constructors $0:\NNO$ and $1{+}\_: \NNO\fun\NNO$; see
  \agdalink{Prelude.Nat}.

  We write $a\equiv a'$ for the \emph{identity type} associated with
  terms $a$ and $a'$ of the same (implicit) type. It is a data type
  that has a unique constructor, called $\refl$, just when $a$ and
  $a'$ are definitionally equal (i.e.~equal according to the
  rules for equality judgements); see
  \agdalink{Prelude.Identity}. Given $a,a':A$, if there is a term of
  type $a\equiv a'$ we say that $a$ and $a'$ are \emph{propositionally
    equal} terms of type $A$. Identity types play an important role
  even though typing constructs in general do not behave extensionally
  with respect to them. In \IRU, propositional equality is in general
  a weaker relation than definitional equality, whereas in extensional
  type theory the two notions coincide by virtue of equality
  reflection (see Figure~\ref{fig:rulee}).
\end{nota}
  
\section{Displayed Setoids} 
\label{sec:diss}

In set theory the word ``setoid'' has an unambiguous meaning: it is a
set equipped with an equivalence relation and is sometimes called a
``Bishop set'' in a constructive setting, because of its use
by~\citet{BishopE:fouca}. But in type theory there are different things
that have been labelled ``setoid'', depending upon what kind of values
the equivalence relation takes: types, or pro\-positions (impredicative
or predicative, strict or non-strict).\footnote{Sometimes
  \emph{partial} equivalence relations (symmetric and transitive, but
  not necessarily reflexive relations) are used, for example by
  \citet{HofmannM:extcit}; but calling such things ``setoids'' rather
  than ``PERs'' is thankfully rare.} Here we take the simplest
possible approach in \IRU, avoiding any notion of proposition (even
the ``mere'' propositions of the~\citet[Definition~3.3.1]{HoTT}) and just
consider equivalence relations that are valued in $\Set$ (defined on
types that are also in $\Set$).

\begin{defn}[\agdalink{Setoid.Definition}]
  \label{def:setd}
  A \emph{setoid} $A$ in \IRU{} is a type $|A|:\Set$ equipped with an
  ITU \emph{equivalence relation}, which is by definition a function
  \begin{equation}\label{eq:1}
    A\ni\_{\eq}\_ : |A|\fun |A| \fun \Set 
  \end{equation}
  together with terms:
  \begin{align}
    \rfl\,A &: \forall x \fun (A\ni x \eq x) \label{eq:2}\\
    \sym\,A &: \forall x\,y\fun (A\ni x \eq y) \fun (A\ni y \eq
    x) \label{eq:3}\\ 
    \trs\,A &: \forall x\,y\,z \fun (A\ni x \eq y) \fun (A\ni y \eq z)
    \fun (A\ni x \eq z) \label{eq:4}
  \end{align}
  $\Setd$ denotes the record type in the universe $\SetOne$ of all
  such setoids. Given $A, A' : \Setd$, a \emph{morphism of setoids $f$
    from $A$ to $A'$} is given by a function
  \begin{equation}\label{eq:5}
    |f|: |A|\fun |A'|
  \end{equation}
  together with a term
  \begin{equation}\label{eq:6}
    \cng\,f : \forall x\,y\fun (A\ni x \eq y) \fun (A'\ni |f|\,x
    \eq |f|\,y)
  \end{equation}
  The record type in $\Set$ of all morphisms from $A$ to $A'$ is the
  underlying type of a setoid denoted $A\morphism A'$, whose equivalence
  relation is defined by
  \begin{equation}
    \label{eq:7}
    ({A\fun A'}\ni f \eq g) = \forall x \fun (A' \ni |f|\,x \eq |g|\,x)
  \end{equation}
  Note that the terms of types \eqref{eq:2}--\eqref{eq:4} and
  \eqref{eq:6} that are part of the structure of setoids and their
  morphisms are not required to satisfy any coherence properties (such
  as group\-oid laws or functor­iality); from the point of view of
  homotopy type theory~\citep{HoTT} one might call the elements of
  $\Setd$ ``wild'' setoids~\citep{CapriottiP:unihcc}.
\end{defn}

Setoid morphisms can be composed in an obvious way, composition is
associative up to definitional equality and has identity morphisms
that are unitary up to definitional equality; see
\agdadef{Setoid.Definition}{HomProperties}{HomProperties}.

In order to model dependent types using setoids in \IRU, given
$A:\Setd$ we need a notion of ``setoid displayed over $A$'' (also
known as ``$A$-indexed family of setoids''). At the very least such a
displayed setoid $B$ should have an underlying displayed type
$\| B \| : | A | \fun \Set$ over the underlying type
$|A|$. Furthermore, we need to have a comprehension operation turning
$B$ displayed over $A$ into an ordinary setoid $A\ltimes B$ equipped
with a morphism from $A \ltimes B$ to $A$ (and other things); and it
is reasonable to expect that this lies over
$\pi_1:(\sum[x \in |A|]\,\|B\|x) \fun |A|$. In which case, to make
$A\ltimes B$ into a setoid we need to give an equivalence relation $R$
on $\sum[x \in |A|]\, \|B\|x$ that is mapped by $\pi_1$ to
$A\ni\_{\eq}\_\,$.  It is tempting to assume that for each
$(x,y),(x', y'): \sum[x \in |A|]\, \|B\|x$ the type
$R\,(x,y)\,(x', y') : \Set$ takes the form of a dependent product
\begin{equation}
  \label{eq:8}
  R\,(x,y)\,(x', y') = \textstyle\sum[p \in (A\ni x \eq x')]\, R_p\,
  y\,y' 
\end{equation}
where $R_p$ relates $y: \|B\|x$ and $y':\|B\|x'$, \emph{dependent upon
  a proof} $p: (A\ni x \eq x')$, in other words a kind of
``path-over'' heterogeneous equality familiar from homotopy type
theory~\citep[section~6.2]{HoTT}. But to fit with the lack of any
coherence assumptions in the definition of setoid, one has to ensure
that $R_p$ is proof irrelevant, that is, does not really depend upon
$p$. One way of doing that is to make this proof irrelevance an
explicit assumption; \citet{PalmgrenE:typtsb} does this in his
Definition~3.8, also giving a number of other coherence conditions
that altogether yield what from our point of view is an unnecessarily
complicated notion of displayed setoid. Another way is to assume a
richer type theory than \IRU, namely one with a universe of
definitionally proof-irrelevant propositions (called \verb+Prop+ by
\citet{agda} and \verb+SProp+ by \citet{rocq}), and to use a notion of
setoid whose equivalence relation component takes its values in that
universe of propositions; for then $R_p$ and $R_{p'}$ are
definitionally equal because $p$ and $p'$ are. This is the route taken
by \citet{AltenkirchT:exteit} and the work that has flowed from
that. Here we will see that it is possible to stick with the very
simple type theory \IRU{} and its ``wild'' setoids, by replacing the
dependent product in \eqref{eq:8} with a cartesian product:
\[
  R\,(x,y)\,(x', y') = (A\ni x \eq x') \times (B\ni{x,y}\heq{x', y'})
\]
where $B\ni\_{\heq}\_$ is a function of type
$(\sum[x \in |A|]\, \|B\|x) \fun (\sum[x \in |A|]\, \|B\|x) \fun\Set$
with properties to be determined. This apparently rather na{\"i}ve way
of dealing with the needed proof-irrelevance has been used recently by
\citet{PujetL:indus}. In fact it already occurs in \citet[Definition
5.3.4]{HofmannM:extcit}, unremarked and as part of a more complicated
model construction that is driven by the desire for computational
content mentioned in the Introduction (and which ``has other
drawbacks''~\citep[page~218]{HofmannM:extcit}, causing it to be
abandoned by its author in future work).

There is explicitly no dependence of $B\ni{x,y}\heq{x', y'}$ on proofs
of $A\ni x \eq x'$ and unlike heterogeneous equality (the ``John
Major'' equality of \citet[section~5.1.3]{McBrideC:deptfp}), in
general there will be no way of constructing a term of type
$A\ni x \eq x'$ from a term of type $B\ni{x,y}\heq{x', y'}$. As a
result, one has to be careful that $B\ni{x,y}\heq{x', y'}$ is only
used in contexts where one already has $A\ni x \eq x'$; the properties
of symmetry~\eqref{eq:12} and transitivity~\eqref{eq:13} given below
are examples of how this can lead to slightly more complicated
definitions than one might first guess ($\hsym\,B$ takes an extra
argument of type $A\ni x\eq x'$ and $\htrs\,B$ takes extra arguments
of types $A\ni x\eq x'$ and $A\ni x'\eq x''$). As well as reflexivity,
symmetry and transitivity, we also have to ensure that the notion of
displayed setoid can model the crucial conversion rule of dependent
type theory
\begin{equation}
  \label{eq:144}
  (\Gamma\ent t:T) \fun (\Gamma\ent T = T') \fun (\Gamma\ent t:
  T')
\end{equation}
This is the role of \eqref{eq:14} and \eqref{eq:15} in the
following definition (we have adopted the names $\coe$ and $\coh$
from \citet[section~3.1]{AltenkirchT:conusm}).

\begin{defn}[\agdalink{Setoid.Display}]
  \label{def:dsetd}
  Given a setoid $A: \Setd$, a \emph{setoid $B$ displayed over $A$} is
  a function $\| B \| : |A| \fun \Set$ equipped with an ITU
  \emph{heterogeneous equivalence relation} over $A\ni\_{\eq}\_$,
  which is by definition a function
  \begin{equation}
    (B\ni\_{\heq}\_) : (\textstyle\sum[x \in |A|]\, \|B\|x) \fun
    (\textstyle\sum[x \in |A|]\, \|B\|x) \fun\Set \label{eq:10}
  \end{equation}
  together with terms
  \begin{align}
    \hrfl\,B
    &: \forall x\,y\fun (B\ni x,y \heq x,y)\label{eq:11}\\
    \hsym\,B
    &:
    \begin{array}[t]{@{}l}
      \forall\{x\,x'\,y\,y'\} \fun (A \ni x \eq x') \fun{}\\
      (B \ni x,y \heq x',y') \fun (B\ni x',y'\heq x,y)
    \end{array}\label{eq:12}\\
    \htrs\,B
    &:
    \begin{array}[t]{@{}l}
      \forall\{ x\,x'\,x''\,y\,y'\,y''\} \fun (A \ni x \eq x') \fun
      (A \ni x' \eq x'') \fun {}\\
      (B \ni x,y \heq x',y') \fun (B \ni
            x',y' \heq x'',y'') \fun (B \ni x,y \heq
      x'',y'')
    \end{array}\label{eq:13}\\
    \coe\,B
    &: \forall\{x\,x'\} \fun (A \ni x \eq x') \fun \|B \| x \fun \| B
      \| x' \label{eq:14}\\
    \coh\,B
    &: \forall\{x\,x'\} \fun (e : A \ni x \eq x')(y : \|B \|) \fun
      (B\ni x,y \heq x', \coe\,B\,e\,y)\label{eq:15}  
  \end{align}
  We write $\Setd[ A ]$ for the record type in $\SetOne$ of setoids
  displayed over $A$. A \emph{section of the
    displayed setoid $B$ over $A$} called $f$ is given by a dependently
  typed function
  \begin{equation}
    \label{eq:16}
    \| f \| : (x : | A |) \fun \| B \| x
  \end{equation}
  together with a term
  \begin{equation}
    \label{eq:17}
    \hcng\, f : \forall x\, x' \fun (A \ni x \eq x') \fun (B \ni x ,\|
    f\|x \heq x' , \| f \| x')
  \end{equation}
  We write $\Setd[ A \Vdash B ]$ for the record type in $\Set$ of all
  sections of  $B$ over $A$. 
\end{defn}

\begin{remk}[\textbf{Fibres of a displayed setoid}]
  \label{rem:fibds}
  If $A: \Setd$ and $B:\Setd[ A]$, then for each $x:A$
  there is a setoid $B\fib x$ with
  \begin{gather}
    | B\fib x | = \| B \|\,x\label{eq:20}\\
    ({B\fib x}\ni y_1 \eq y_2) = B\ni x,y_1 \heq x,y_2\label{eq:21}
  \end{gather}
  If $e : (A \ni x_1 \eq x_2)$ then we get a setoid morphism
  $B\mathrel{''}e : | B\fib x_1 \morphism B\fib x_2 |$
  with
  \begin{equation}
    \label{eq:23}
    |B\ffib e| = \coe\,B\,e
  \end{equation}
  (using $\rfl\,A$, $\sym\,A$, $\hsym\,B$, $\htrs\,B$, $\coe\,B$, and
  $\coh\,B$ to prove the congruence property~\eqref{eq:6}); and this
  morphism is independent of the proof of equality of $x_1$ and $x_2$
  in $A$, in the sense that if $e_1,e_2: (A \ni x_1 \eq x_2)$ then
  there is a term of type
  $({B\fib x_1}\morphism{B\fib x_2}) \ni B\ffib e_1
  \eq B\ffib e_2$. For the details see
  \agdadef{Setoid.Display}{Fibres}{Fibres}.
\end{remk}

\begin{defn}[\textbf{Re-indexing}]
  \label{def:rei}
  Displayed
  setoids and their sections can be \emph{re-indexed} along setoid
  morphisms: if $A,A': \Setd$ and $f : | A' \morphism A |$, then
  for each $B:\Setd[ A]$ there is a displayed setoid $f \ast B : \Setd[
  A' ]$ with
  \begin{equation}
    \label{eq:126}
    \begin{aligned}
      &\| f \ast B \| = \| B \| \comp |f|\\
      &({f \ast B}\ni x,y\heq x', y') = (B\ni |f|\,x, y \heq
      |f|\,x', y')
    \end{aligned}
  \end{equation}
  and for $b: \Setd[ A \Vdash B ]$ there is a section\footnote{We
    write $f \ast_B b$ rather than $f\ast b$ because the latter
    notation is too ambiguous for Agda to be able to infer in general
    which display $B$ is involved.}
  $f \ast_B b : \Setd[ A' \Vdash f \ast B ]$ with
  \begin{equation}
    \label{eq:127}
    \| f \ast_B b \| = \| b \| \comp \|f\|
  \end{equation}
  See \agdadef{Setoid.Display}{ReIndex}{ReIndex} for the details,
  including the fact that re-indexing is associative and unitary up to
  definitional equality.
\end{defn}

\begin{defn}[\textbf{Comprehension and product}]
  \label{def:com}
  Given $A:\Setd$, there is a \emph{comprehension} operation that turns
  a displayed setoid $B:\Setd[ A ]$ into a setoid $A \ltimes B :
  \Setd$ with
  \begin{equation}
    \label{eq:18}
    \begin{aligned}
    &|A \ltimes B | = \textstyle\sum[ x \in | A |]\, \| B \|
    x\\
    &({A \ltimes B} \ni (x, y) \eq (x' , y')) =
    (A \ni x \eq x') \times (B\ni x,y\heq x',y')
  \end{aligned}
  \end{equation}
  See \agdadef{Setoid.Display}{Comprehension}{Comprehension} for more
  details, including the projection and pairing operations associated
  with comprehension. We will also need the simpler, non-dependent
  notion of the product of two setoids $A,A':\Setd$, which we denote
  by $A\otimes A'$. This satisfies
  \begin{equation}
    \label{eq:19}
    \begin{aligned}
      &|A \otimes A'| = |A| \times |A'|\\
      &(A \otimes A' \ni (x, y) \eq (x' , y')) =
    (A \ni x \eq x') \times (A'\ni y\eq y')
    \end{aligned}
  \end{equation}
  See \agdalink{Setoid.Definition}.
\end{defn}

Altogether setoids, displayed setoids and their sections have the
structure of a \emph{category with families} (CwF) of
\citet{DybjerP:intt} in a very strict sense: all the necessary
identities are satisfied up to definitional equality in Agda. However
this is not the CwF we will use to model extensional type theory, for
two reasons. First, it does not provide a model all the relevant
typing constructs up to definitional equality. For example, although
the structure needed to interpret $\Pi$-types can be given, it only
satisfies the $\eta$-rule up to the equivalence given by an
appropriate setoid; and something similar is true for identity types,
whose computation rule only holds up to the appropriate setoid
equivalence relation.\footnote{This is true even for the CwFs
  constructed using definitionally proof irrelevant propositions
  by\citet[page~9]{AltenkirchT:conusm}; recent work of
  \citet{PujetL:revhsm} shows how to overcome such limitations.}
Secondly, we wish to use the universe-oriented version of CwFs of
\citet[section~2.3]{CoquandT:canndt}\footnote{More precisely, a
  non-cumulative variant of that.}, since that makes it easier to give
the semantics of a countable hierarchy of universes of types. That
involves changing the notion of ``family'' from just being a display:
we need a setoid universe at level $n+1$ that contains a code for the
universe at level $n$ in such a way that families of (codes for)
level-$n$ setoids coincide up to definitional equality with sections
of the level-$n$ universe code. We show how to do this in
Section~\ref{sec:fame}.

\section{Setoid Universes}
\label{sec:setu}

The type $\Setd$ from Defintion~\ref{def:setd} can itself be made into
a large setoid, that is, endowed with a type-valued equivalence
relation, in several different ways. For example we could use the
identity type $\_{\equiv}\_$ for $\Setd$; or we could use the notion
of isomorphism determined by morphisms of setoids with their
composition and identity operations. However, none of the known
choices gives rise to a model of a type theory like \ETU{}, the
problem partly being how to interpret \eqref{eq:144} without running
into coherence problems~\citep[p~1284]{PalmgrenE:typtsb}. One could
say that we do not know how to construct an ``open universe'' out of
setoids in a way that fully models type theory with (a hierarchy of)
universes.

Instead one can construct ``closed'' universes, that is, decide in
advance what constructs the universe should contain and then make an
inductive definition of codes for setoids and, simultaneously,
recursive definitions of three things: the setoid denoted by a code,
the equivalence of two codes and the heterogeneous equivalence of
elements of the coded setoids. As discussed in the Introduction,
\citet{AltenkirchT:conusm} and \citet{PalmgrenE:typtsb} both take this
approach, but in different ways. Here we use the approach of
\citet{AltenkirchT:conusm}, as does \citet{PujetL:indus}. Since we want to
model the particular version of extensional type theory described in
section~\ref{sec:extu}, the typing constructs we are fixing in advance
are: dependent products (with $\eta$-conversion), extensional equality
types, empty type, natural number type and a countable hierarchy of
universes closed under those constructs.

To construct the setoid universes we need in Agda, we declare certain
data types $|\U_l|: \Set$ and simultaneously recursively define functions
\begin{gather}
  \| \El_l\| : |\U_l| \fun \Set\label{eq:25}\\
  (\_)\eq_l(\_) : |\U_l| \fun |\U_l| \fun \Set\label{eq:26}\\
  (\_,\_)\heq_l(\_,\_) : (X : |\U_l|) \fun \|\El_l\|X \fun (X' : |\U_l|) \fun
  \|\El_l\|X' \fun \Set\label{eq:27}
\end{gather}
where $l$ ranges over $\NNO$. Then we prove that \eqref{eq:26} is an
\IRU{} equivalence relation (Definition~\ref{def:setd}) and that
\eqref{eq:27} is an \IRU{} heterogeneous equivalence relation over it
(Definition~\ref{def:dsetd}), so that we get a setoid
\begin{gather}
  \U_l:\Setd\label{eq:113}\\
  \intertext{with $\U_l\ni \_{\eq}\_$ given by $\eq_l$, and a
    displayed setoid} 
  \El_l:\Setd[ \U_l ] \label{eq:114}
\end{gather}
with $\El_l\ni\_{\heq}\_$ given by $\heq_l$. Although Agda\footnote{We
  used version~2.8.0} admits such inductive-recursive definitions, its
termination checker is defeated if we try to define the whole
$\NNO$-indexed family by structural recursion on
$l$.\footnote{Furthermore, to satisfy the termination checker it is
  necessary to use a four-place function in \eqref{eq:27} rather than
  a binary function on $\sum[ X \in |\U_l| ]\, \|\El_l\| X$.} So
instead we give the base case ($l=0$) and a successor operation
$(\sum[ U \in \Setd ]\, \Setd[ U ]) \fun (\sum[ U \in \Setd ]\, \Setd[
U ])$, and then define the whole sequence by iteration (a special case
of $\NNO$-elimination).

\subsection{Base universe}

The data type $|\U_0|:\Set$ has four constructors
\begin{gather}
  \Nat : |\U_0|\label{eq:32}\\
  \Emp : |\U_0|\label{eq:9}\\
  \Eq_0 : (A : |\U_0|)(a\,a':\|\El_0\|A) \fun |\U_0|\label{eq:22}\\
  \Pit_0 :
  \begin{array}[t]{@{}l}
    (A : |\U_0|)\\
    (B : \|\El_0\|A \fun |\U_0|)\\
    (q: \forall a\,a'\fun (A,a \heq_0 A', a') \fun (B\,a\eq_0 B\,a'))\\
    {}\fun |\U_0| 
  \end{array}\label{eq:24}
\end{gather}
that code the natural number type, the empty type, equality
types and $\Pi$-types respectively:
\begin{gather}
  \|\El_0\|\,\Nat = \NNO\label{eq:33}\\
  \|\El_0\|\,\Emp = \emptyset\label{eq:34}\\
  \|\El_0\|(\Eq_0\,A\,a\,a') = (A,a\heq_0 A,a')\label{eq:36}\\
  \|\El_0\|(\Pit_0\,A\,B\,q) =
  \begin{array}[t]{@{}l}
    \textstyle\sum[ f \in (a:\|\El_0\|A)\fun \|\El_0\|(B\,a) ]\\
    (\forall a\,a' \fun (A,a\heq_0 A,a') \fun (B\,a, f\,a \heq_0
    B\,a', f\,a')
  \end{array}\label{eq:35} 
\end{gather}
Note that the type \eqref{eq:36} coded by $\Eq_0\,A\,a\,a'$ consists
of equivalences in the fiber of $\El_0$ over $A$
(cf.~Remark~\ref{rem:fibds}); such equivalences will model
definitional equality in type theory and hence we get a model of
\emph{extensional} equality types, whose terms correspond to proofs of
definitional equality. The type \eqref{eq:35} coded by
$\Pit_0\,A\,B\,q$ does not just consist of dependently typed
functions, but rather is essentially a subtype of that (since the
second components of elements of this $\Sigma$-type are never compared
for equality). Using partial equivalence relations (PERs) rather than
equivalence relations, one could simplify $\Pit_0$ by dropping its $q$
argument and take its decoded type to just be a $\Pi$-type, at the
expense of a non-reflexive relation $\heq_0$. For the case of simple
type theory that is essentially what \citet{GandyR:axieI} does,
because there is no other possibility in that setting; but in
dependent type theory with $\Sigma$-types one has the opportunity to
build in reflexivity and only work with equivalences. It may be
possible to develop a PER version of the material presented here
(cf.~\citet{BerryD:forpct}), but the author's experiments with
doing that led him to abandon the attempt in favour of total
equivalence relations, which in a sense fit better with the total
functional programming provided by dependent type theory. I believe
Altenkirch was the first to observe this possibility of using
equivalence relations instead of PERs and in any case definitions
\eqref{eq:24} and \eqref{eq:35} are like the corresponding definitions
of \citet{AltenkirchT:exteit}.

We still have to define $\eq_0$ and $\heq_0$. The non-trivial clauses
in the recursive definition of $\eq_0$ are when its two arguments are
constructor patterns with the same outermost form:
\begin{gather}
  (\Nat\eq_0\Nat) = \top\label{eq:38}\\
  (\Emp\eq_0\Emp) = \top\label{eq:39}\\
  (\Eq_0\,A\,a\,b \eq_0 \Eq_0\,A'\,a'\,b') =
  (A \eq_0 A') \times (A,a\heq_0 A',a') \times (A,b\heq_0
  A',b')\label{eq:40}\\ 
  (\Pit_0\,A\,B\,q \eq_0 \Pit_0\,A'\,B'\,q') =
  (A \eq_0 A') \times \forall a\,a' \fun (A,a\heq_0 A',a') \fun (B\,a
  \eq_0 B'\,a') \label{eq:37} 
\end{gather}
In all other cases (constructor patterns with different outermost
form), we just return $\emptyset$. Similarly, the non-trivial clauses
for $\heq_0$ are ``on the diagonal'' and we return $\emptyset$ in the
other cases:
\begin{gather}
  (\Nat,n\heq_0 \Nat,n') = (n\equiv n') \label{eq:41}\\
  (\Emp,\_\heq \Emp,\_) = \top\label{eq:42}\\
  (\Eq_0\,A\,a\,b, \_ \heq_0 \Eq_0\,A'\,a'\,b', \_) =
  \top\label{eq:43}\\
  (\Pit_0\,A\,B\,q, (f, \_) \heq_0 (\Pit_0\,A'\,B'\,q', (f', \_)) =
  \begin{array}[t]{@{}r}
    \forall a\,a'\fun (A,a\heq_0 A',a') \fun{}\\
    (B\,a, f\,a \heq_0 B'\,a', f'\,a')
  \end{array}\label{eq:44}
\end{gather}
Note that \eqref{eq:43} reflects the fact that in an extensional
equality type all proofs are definitionally equality (cf.~rule
\texttt{UIP} in Figure~\ref{fig:rulee}).  With these definitions one
can show that the required terms $\rfl$, $\sym$, $\trs$, $\hrfl$,
$\hsym$, $\htrs$, $\coe$ and $\coh$ can be defined (the last two have
to be defined simultaneously with $\trs$ and $\htrs$), enabling the
setoid $\U_0$ to be defined with $\U_0\ni\_{\eq}\_$ given by $\eq_0$,
and the displayed setoid $\El_0$ to be defined with
$\El_0\ni\_{\heq}\_$ given by $\heq_0$; see
\agdadef{Setoid.Universes}{BaseUniverse}{BaseUniverse}.

\subsection{Successor-universe operation}

We will be modelling a version of extensional type theory defined
below in section~\ref{sec:extu} that has a countable hierarchy of
universes with (among other things) Agda-style $\Pi$-types. Recall
that if $A$ is a type in universe level $\ell$ in Agda and $B\,x$ is a
type in universe level $\ell'$ for each $x:A$, then $(x:A)\fun B\,x$
is in the universe whose level $\ell''$ is the least upper bound of
$\ell$ and $\ell'$. This suggests that when considering codes for
$\Pi$-types in $\U_{1{+} l}$ we need to consider $\U_m$ for all
$m\leq l$ and hence that $\lambda(l:\NNO) \fun \U_l$ has to be defined
by a course-of-values recursion over $\NNO$. A much simpler solution
is to observe that even if Agda universes are not cumulative (elements
of a universe are not literally elements of higher level universes),
there are lifting operations that inject a universe into a higher one
(see \agdalink{Prelude.Level}); and $(x:A)\fun B\,x$ is isomorphic to
the type obtained by first lifting $A$ and $B$ to the common universe
level $\ell''$ before taking a $\Pi$-type there. Accordingly, by
including codes for lifted types in the setoid universes we can get
away with a much simpler, iterative definition of
$\lambda(l:\NNO) \fun \U_l$.

Given $U:\Setd$ and $E : \Setd[ U ]$ we declare a data type
$|\U_+|:\Set$ that depends implicitly upon them, and simultaneously
define recursive functions
\begin{gather}
  \|\El_+\| : |\U_+| \fun \Set\label{eq:45}\\
  (\_)\eq_+(\_) : |\U_+| \fun |\U_+| \fun \Set\label{eq:46}\\
  (\_,\_)\heq_+(\_,\_) : (X : |\U_+|) \fun \|\El_+\|\,X \fun (X' :
  |\U_+|) \fun 
  \|\El_+\|\,X' \fun \Set\label{eq:47}
\end{gather}
$|\U_+|$ has four constructors
\begin{gather}
  \Univ : |\U_+|\label{eq:48}\\
  \Lft : |U| \fun |\U_+|\label{eq:49}\\
  \Eq_+ : (A : |\U_+|)(a\,a':\|\El_+\|\,A) \fun |\U_+|\label{eq:50}\\
  \Pit_+ :
  \begin{array}[t]{@{}l}
    (A : |\U_+|)\\
    (B : \|\El_+\|\,A \fun |\U_+|)\\
    (q : \forall a\,a'\fun (A,a \heq_+ A', a') \fun (B\,a\eq_+ B\,a'))\\
    {}\fun |\U_+| 
  \end{array}\label{eq:51}
\end{gather}
that code $U$, injecting $U$ into $\U_+$, equality types and
$\Pi$-types respectively:
\begin{gather}
  \|\El_+\|\,\Univ = |U|\label{eq:52}\\
  \|\El_+\|(\Lft\,A) = \| E \|\,A\label{eq:53}\\
  \|\El_+\|(\Eq_+\,A\,a\,a') = (A,a\heq_+ A,a')\label{eq:75}\\
  \|\El_+\|(\Pit_+\,A\,B\,q) =
  \begin{array}[t]{@{}l}
    \textstyle\sum[ f \in ((a:\|\El_+\|\,A)\fun \|\El_+\|(B\,a)) ]\\
    (\forall a\,a' \fun (A,a\heq_+ A,a') \fun (B\,a, f\,a \heq_+
    B\,a', f\,a')
  \end{array}\label{eq:72}
\end{gather}
The non-trivial clauses in the definitions of $\eq_+$ and $\heq_+$ are
\begin{gather}
  (\Univ \eq_+ \Univ) = \top\label{eq:54}\\
  (\Lft\,A\eq_+ \Lft A') = (U \ni A \eq A')\label{eq:55}\\
  (\Eq_0\,A\,a\,b \eq_0 \Eq_0\,A'\,a'\,b') =
  \ldots\text{like \eqref{eq:40}}\ldots\label{eq:56}\\
  (\Pit_+\,A\,B\,q \eq_0 \Pit_+\,A'\,B'\,q') =
  \ldots\text{like \eqref{eq:37}}\ldots\label{eq:57}\\
  (\Univ,A \heq_+ \Univ,A') = (U \ni A \eq A')\label{eq:58}\\
  (\Lft\,A, a\heq_+ \Lft\,A',a') = (E \ni A,a\heq
  A',a')\label{eq:59}\\
  (\Eq_+\,A\,a\,b, \_ \heq_+ \Eq_+\,A'\,a'\,b', \_) =
  \ldots\text{like \eqref{eq:43}}\ldots\label{eq:60}\\
  (\Pit_+\,A\,B\,q, (f, \_) \heq_+ (\Pit_+\,A'\,B'\,q', (f', \_)) =
  \ldots\text{like \eqref{eq:44}}\ldots\label{eq:61}
\end{gather}
with $\emptyset$ being the value in all other cases. Once again these
definitions allow the required terms $\rfl$, $\sym$, $\trs$, $\hrfl$,
$\hsym$, $\htrs$, $\coe$ and $\coh$ to be constructed (see
\agdadef{Setoid.Universes}{SucessorUniverse}{SucessorUniverse}), enabling us to
  define
\begin{gather}
  \U_+ : \{U: \Setd\}\{E : \Setd[ U ]\} \fun
  \Setd \label{eq:62}\\
  \El_+ : \{U : \Setd\}\{E : \Setd[ U ]\} \fun
  \Setd[ \U_+\{U\}\{E\} ] 
\end{gather}
with $\U_+\ni\_{\eq}\_$ given by $\eq_+$ and $\El_+\ni\_{\heq}\_$
given by $\heq_+$.

Now we can define the hierarchy of setoid universes \eqref{eq:113} and
generic displays \eqref{eq:114} by starting with $(\U_0,\El_0)$ and
iterating the successor operation
\[
  \U_{1{+} l} = \U_+\{\U_l\}\{\El_l\}
  \qquad
  \El_{1{+} l} = \El_+\{\U_l\}\{\El_l\}.  
\]

\subsection{A setoid universe for contexts}
\label{sec:setuc}

We have not quite finished constructing setoid universes.  Types in
extensional type theory will be modelled not by setoids themselves but
by codes for setoids, for which we have a well behaved notion of
equivalence that is used to model definitional equality of types.
Typing contexts also need to be modelled by a setoid of codes with a
well-behaved notion of equivalence, since we have to give a semantics
to definitional equality of contexts (even though it is a derivable
notion in the presentation of section~\ref{sec:extu}). Therefore we
will construct a setoid universe for codes of typing contexts. It just
needs a unit type for the empty context and a $\Sigma$-type for
extending a context with a type (at any universe level). We call it
\begin{gather}
  \C:\Setd\label{eq:115}\\
  \intertext{ and the associated display}
  \E:\Setd[ \C ]\label{eq:116}
\end{gather}
We define a data type $|\C|:\Set$ and simultaneously define recursive
functions
\begin{gather}
  \|\E\| : |\C|\fun \Set \label{eq:68}\\
  \_{\eq^\c}\_ : |\C| \fun |\C| \fun \Set\label{eq:67}\\
  \_,\_{\heq^\c}\_,\_ : (C : |\C|) \fun \|\E\|\,C \fun
  (C' : |\C|) \fun \|\E\|\,C' \fun \Set\label{eq:69}
\end{gather}
$|\C|$ has two constructors
\begin{gather}
  \Unit : |\C|\label{eq:70}\\
  \Sigt :
  \begin{array}[t]{@{}l}
    (C : |\C|)(l : \NNO)(F : \|\E\|\,C \fun |\U_l|)\\
    \left(q : \forall c\,c'\fun (C,c \heq^\c C,c') \fun
    (\U_l\ni F\,c \eq F\,c')\right)
    \fun |\C| 
  \end{array} \label{eq:71}
\end{gather}
that code the empty context and context extension respectively:
\begin{gather}
  \|\E\|\,\Unit = \top\label{eq:73}\\
  \|\E\|\,(\Sigt\,C\,l\,F\,q) =
  \textstyle\sum[ c \in \|\E\|\,C ]\,\|
  \El_l\|(F\,c)\label{eq:74} 
\end{gather}
The non-trivial clauses in the definitions of $\eq^\c$ and
$\heq^\c$ are
\begin{gather}
  (\Unit \eq^\c \Unit) = \top\label{eq:76}\\
  (\Sigt\,C\,l\,F\,q \eq^\c \Sigt\,C'\,l'\,F'\,q') =(C \eq^\c
  C') \times \textstyle\sum[ p \in (l \equiv l')]\label{eq:77}\\
    \qquad\forall c\,c' \fun (C,c\heq^\c C',c') \fun (\U_{l'}\ni
  \mathtt{subst}\,\_\,p\,(F\,c) \eq F'\, c')\notag\\
  (\Unit , \tt \heq^\c \Unit,\tt) = \top\label{eq:78}\\
  (\Sigt\,C\,l\,F\,q , (c, t) \heq^\c \Sigt\,C'\,l'\,F'\,q' , (c',
  t')) = {}\label{eq:79}\\
    \qquad (C,c \heq^\c C',c') \times
    \textstyle\sum[ p \in (l \equiv l')]\,
    (\El_{l'}\ni \mathtt{subst}\,\_\,p\,(F\,c, t) \heq (F'\,c',t'))
    \notag 
\end{gather}
with $\emptyset$ being the value in all other cases. (Clauses
\eqref{eq:77} and \eqref{eq:79} use the function $\mathtt{subst}$ for
transporing Agda terms along identity proofs; see
\agdalink{Prelude.Identity}.) As before, these definitions allow the
required terms $\rfl$, $\sym$, $\trs$, $\hrfl$, $\hsym$, $\htrs$,
$\coe$ and $\coh$ to be constructed (see
\agdadef{Setoid.Universes}{ContextUniverse}{ContextUniverse}),
enabling us to define \eqref{eq:115} with $\C\ni \_{\eq}\_$ given by
$\eq^\c$, and \eqref{eq:116} with $\E\ni\_{\heq}\_$ given by
$\heq^\c$.

\section{Extensional Type Theory with Universes}
\label{sec:extu}

We give an ``extrinsic'' formalization in \IRU{} of a version of
the extensional type theory with
universes of~\citet{Martin-LoefP:inttt}, which we call \ETU.  Thus there are
\IRU{} inductive datatypes of raw \ETU{} expressions and inductively
defined judgements about which expressions are well-formed. We use
three basic forms of judgement (see \agdalink{ETU.Judgement})
\begin{align}
  &\Ok\Gamma
  &&\text{context $\Gamma$ is well-formed}\label{eq:84}\\
  &\Gamma\ent a :_l A
  &&\text{in context $\Gamma$, term $a$ has type $A$ in universe
     $\ETUU_l$}\label{eq:85}\\ 
  &\Gamma \ent a=a':_l A
  &&\text{in context $\Gamma$, terms $a$ and $a'$ are definitionally
     equal}\notag\\
  &&&\text{of type $A$ in universe $\ETUU_l$}\label{eq:86} 
\end{align}
where $l:\NNO$ is a universe level.  Judgements about being a
well-formed type, or being definitionally equal types are special
cases of these, since we only consider types that lie in some universe
and the universes $\ETUU_l$ are themselves types (indeed, if
$\Ok\Gamma$ then $\Gamma\ent \ETUU_l :_{2{+} l} \ETUU_{1{+} l}$). We
make the following abbreviations for judgements about types:
\begin{equation}\label{eq:111}
  \begin{array}{rcl}
    (\Gamma\ent A : \ETUU_l) &=& (\Gamma\ent A:_{1{+} l}\ETUU_l)\\
    (\Gamma\ent A = A' : \ETUU_l) &=& (\Gamma \ent A = A' :_{1{+} l}\ETUU_l)
  \end{array}
\end{equation}
The syntax of \ETU{} terms (some of which denote types) and contexts use the
following notation (see \agdalink{ETU.Syntax}):
\begin{align}
  &\ETUU_l &&\text{universe type}\label{eq:87}\\
  &\ETUPi_{l,l'}\,A\,B &&\text{$\Pi$-types}\label{eq:88}\\
  &\ETUlam_A\,b &&\text{function abstraction}\label{eq:89}\\
  &\ETUapp{b}{A}{B}{a} &&\text{function application}\label{eq:90}\\
  &\ETUEq_A\,a\,a' &&\text{extensional equality type}\label{eq:91}\\
  &\ETUrefl_A\,a &&\text{reflexivity proof}\label{eq:92}\\
  &\ETUEmp &&\text{empty type}\label{eq:93}\\
  &\ETUemp_A\,a &&\text{empty type eliminator}\label{eq:94}\\
  &\ETUNat &&\text{natural number type}\label{eq:95}\\
  &\ETUzero &&\text{zero}\label{eq:96}\\
  &\ETUsucc\,a &&\text{successor}\label{eq:97}\\
  &\ETUnrec_B\,b\,b'\,a &&\text{eliminator for natural number
                           type}\label{eq:98}\\ 
  &x &&\text{variable ($x:\Atom$)}\label{eq:101}\\
  &\ETUECx &&\text{empty context}\label{eq:99}\\
  &\Gamma\cons{x:_l A} &&\text{extended context}\label{eq:100}
\end{align}
Note that we are giving a ``core'' language rather than a ``surface''
language (to use the terminology of \citet{SterlingJ:bidec}), in which
terms are decorated with explicit universe levels and types sufficient
for them to be assigned a meaning by the semantics of
section~\ref{sec:semetu}.

To deal with variable binding in the syntax we employ the well-scoped
locally nameless representation of syntax advocated
by~\citet{PittsAM:welsln} and its accompanying Agda library
\agdalink{WSLN}. Thus for each $n:\NNO$ there is an Agda datatype
$\Tm{n}$ of $n$-fold abstracted terms, or \emph{$n$-terms} for short;
These $n$-terms may contain named free variables (Agda terms of type
$\Atom$, which is a copy of $\NNO$) and de~Bruijn indices less than
$n$ (which are terms of the datatype $\Fin\,n$ of finite ordinals less
than $n$). For example in \eqref{eq:88} $B$ is a $1$-term, as is $b$
in \eqref{eq:89}; and in \eqref{eq:98} $B$ is a $1$-term and $b'$ is a
$2$-term (see Figure~\ref{fig:nrec} to see why). This form of syntax
representation factors out $\alpha$-equivalence while remaining close
to informal, ``nameful'' accounts. Given $x:\Atom$ and $t: \Tm{n}$
there is an \emph{abstraction} $x.\,t:\Tm{1{+} n}$ and given
$t:\Tm{1{+} n}$ and $u:\Tm{0}$ there is a \emph{concretion}
$t[ u ]: \Tm{n}$; the definition of these operations involves some
de~Bruijn index yoga (see \citet[section~4]{PittsAM:welsln}), but once
defined they have the expected properties and can be used without much
pain. Most of the time we only need $0$-terms, which we just refer to
as \emph{terms} and write $\TM$ for $\Tm{0}$; indeed the expressions
$a$, $a'$ and $A$ in the judgements \eqref{eq:85} and \eqref{eq:86}
are $0$-terms.  Some terms denote types and we write $T:\Ty{n}$ (or
$T:\TY$ in case $n=0$) when we want to flag that an $n$-term $T$ is
playing the role of a type.

The rules for inductively defining the valid \ETU{} judgements are given
in \agdalink{ETU.Rules}. More specifically, the judgements
\eqref{eq:84}--\eqref{eq:86} are Agda datatypes and each rule takes
the form of a constructor for one of those datatypes. Here we only
reproduce a few of them for comment.

\begin{figure}
  \renewcommand{\arraystretch}{.8}
  \begin{mathpar}
    \begin{array}[b]{l}
      {\ent}{\ETUEq}:\\
      \quad\{l:\NNO\}\\
      \quad \{A\,a\,b: \TM\}\\
      \quad (\_ : \Gamma \ent a:_l A)\\
      \quad (\_ : \Gamma \ent b :_l A)\\
      \quad\text{{-}\,{-} helper hypothesis}\\
      \quad (\_ : \Gamma \ent A : \ETUU_l)\\
      \quad{\fun}\;\text{{-}\,{-}\,{-}\,{-}\,{-}\,%
      {-}\,{-}\,{-}\,{-}\,{-}\,{-}\,{-}\,{-}\,{-}\,{-}}\\
      \quad\Gamma \ent \ETUEq_A\,a\,b :\ETUU_l
    \end{array}
    \and
    \begin{array}[b]{l}
      {\ETUEq}\mathtt{Cong}:\\
      \quad\{l:\NNO\}\\
       \quad \{A\,A'\,a\,a'\,b\,b': \TM\}\\
       \quad(\_ : \Gamma \ent A = A' : \ETUU_l)\\
      \quad (\_ : \Gamma \ent a = a':_l A)\\
      \quad (\_ : \Gamma \ent b = b':_l A)\\
      \quad{\fun}\;\text{{-}\,{-}\,{-}\,{-}\,{-}\,{-}%
      \,{-}\,{-}\,{-}\,{-}\,{-}\,{-}\,{-}\,{-}\,{-}\,{-}\,{-}\,{-}\,{-}\,{-}}\\
      \quad\Gamma \ent \ETUEq_A\,a\,b = \ETUEq_{A'}\,a'\,b' :\ETUU_l 
    \end{array}
    \\
    \begin{array}[b]{l}
      {\ent}{\ETUrefl}:\\
      \quad\{l:\NNO\}\\
      \quad \{A\,a: \TM\}\\
      \quad (\_ : \Gamma \ent a:_l A)\\
      \quad\text{{-}{-} helper hypothesis}\\
      \quad (\_ : \Gamma \ent A :\ETUU_l)\\
      \quad{\fun}\; \text{{-}\,{-}\,{-}\,{-}\,{-}%
      \,{-}\,{-}\,{-}\,{-}\,{-}\,{-}\,{-}\,{-}\,{-}\,{-}}\\
      \quad\Gamma \ent \ETUrefl_A\,a :_l \ETUEq_A\,a\,a
    \end{array}
    \and
    \begin{array}[b]{l}
      {\ETUrefl}\mathtt{Cong}:\\
      \quad\{l:\NNO\}\\
      \quad \{A\,A'\,a\,a': \TM\}\\
      \quad (\_ : \Gamma \ent A = A' :\ETUU_l)\\
      \quad(\_ : \Gamma \ent a = a' :_l A)\\
      \quad{\fun}\; \text{{-}\,{-}\,{-}\,{-}\,{-}\,{-}\,{-}%
      \,{-}\,{-}\,{-}\,{-}\,{-}\,{-}\,{-}\,{-}\,{-}\,{-}%
      \,{-}\,{-}\,{-}\,{-}\,{-}}\\
      \quad\Gamma \ent \ETUrefl_A\,a : \ETUrefl_{A'}\,a' :_l \ETUEq_A\,a\,a
    \end{array}
    \\
    \begin{array}[b]{l}
     \mathtt{Reflect}:\\
      \quad\{l:\NNO\}\\
      \quad \{A\,a\,b\,e: \TM\}\\
      \quad (\_ : \Gamma \ent a:_l A)\\
      \quad (\_ : \Gamma \ent b :_l A)\\
      \quad(\_ : \Gamma \ent e :_l \ETUEq_A\,a\,b)\\
      \quad\text{{-}\,{-} helper hypothesis}\\
      \quad (\_ : \Gamma \ent A :\ETUU_l)\\
      \quad{\fun}\;\text{{-}\,{-}\,{-}\,{-}\,{-}\,{-}\,{-}%
      \,{-}\,{-}\,{-}\,{-}\,{-}\,{-}\,{-}\,{-}}\\
      \quad\Gamma \ent a = b :_l A
    \end{array}
    \and
    \begin{array}[b]{l}
     \mathtt{UIP}:\\
      \quad\{l:\NNO\}\\
      \quad \{A\,a\,b\,e\,e': \TM\}\\
      \quad (\_ : \Gamma \ent a :_l A)\\
      \quad (\_ : \Gamma \ent b :_l A)\\
      \quad(\_ : \Gamma \ent e :_l \ETUEq_A\,a\,b)\\
      \quad(\_ : \Gamma \ent e' :_l \ETUEq_A\,a\,b)\\
      \quad\text{{-}\,{-} helper hypothesis}\\
      \quad (\_ : \Gamma \ent A :\ETUU_l)\\
      \quad{\fun}\;\text{{-}\,{-}\,{-}\,{-}\,{-}\,{-}\,{-}%
      \,{-}\,{-}\,{-}\,{-}\,{-}\,{-}\,{-}\,{-}}\\
      \quad\Gamma \ent e = e' :_l \ETUEq_A\,a\,b
    \end{array}
    \\
    \begin{array}[b]{l}
      {\ent}{\ETUEmp}:\\
      \quad(\_ : \Ok\,\Gamma)\\
      \quad{\fun}\;\text{{-}\,{-}\,{-}\,{-}\,{-}\,{-}\,{-}\,{-}\,{-}}\\
      \quad \Gamma\ent \ETUEmp : \ETUU_0
    \end{array}
    \and
    \begin{array}[b]{l}
      {\ent}{\ETUemp}:\\
      \quad\{l:\NNO\}\\
      \quad\{a\,B : \TM\}\\  
      \quad(\_ : \Gamma\ent B :\ETUU_l)\\
      \quad(\_ : \Gamma \ent a :_0 \ETUEmp)\\
      \quad{\fun}\;\text{{-}\,{-}\,{-}\,{-}\,{-}\,{-}%
      \,{-}\,{-}\,{-}\,{-}\,{-}\,{-}}\\
      \quad \Gamma\ent \ETUemp_B\, a :_l B
    \end{array}
    \and
      \begin{array}[b]{l}
      {\ETUemp}\mathtt{Cong}:\\
      \quad\{l:\NNO\}\\
      \quad\{a\,a'\,B\,B' : \TM\}\\  
      \quad(\_ : \Gamma\ent B = B' :\ETUU_l)\\
      \quad(\_ : \Gamma \ent a = a' :_0 \ETUEmp)\\
        \quad{\fun}\;\text{{-}\,{-}\,{-}\,{-}\,{-}\,{-}\,{-}%
        \,{-}\,{-}\,{-}\,{-}\,{-}\,{-}\,{-}\,{-}\,{-}\,{-}\,{-}\,{-}\,{-}}\\
      \quad \Gamma\ent \ETUemp_B\, a = \ETUemp_{B'}\, a' :_l B
      \end{array}
      \\
    \end{mathpar}
    
  \caption{Rules for extensional equality and the empty type}
  \label{fig:rulee}
\end{figure}

\begin{figure}
  \renewcommand{\arraystretch}{.8}
  \begin{mathpar}
    \begin{array}[b]{l}
      {\ent}{\ETUnrec}:\\
      \quad\{l:\NNO\}\\
      \quad\{B : \Ty{1}\}\\
      \quad\{b\,a: \TM\}\\
      \quad\{b' : \Tm{2}\}\\
      \quad(S :\FsetA)\\
      \quad(\_: \Gamma \ent b :_l B [\ETUzero])\\
      \quad(\_ : \forall x\,y \fun x\freshfor y\freshfor S \fun
      \Gamma\cons {x :_0 \ETUNat}\cons {y :_l B[ x ]} \ent
      b'[x][y] :_l B[ \ETUsucc\,x ])\\
      \quad(\_ : \Gamma \ent a :_0 \ETUNat)\\
      \quad\text{{-}\,{-} helper hypothesis}\\
      \quad(\_ : \forall x \fun x\freshfor S \fun \Gamma\cons {x:_0
      \ETUNat} \ent B[x] :\ETUU_l)\\
      \quad{\fun}\;\text{{-}\,{-}\,{-}\,{-}\,{-}\,{-}\,{-}\,{-}%
      \,{-}\,{-}\,{-}\,{-}\,{-}\,{-}\,{-}\,{-}\,{-}\,{-}\,{-}%
      \,{-}\,{-}\,{-}\,{-}\,{-}\,{-}\,{-}\,{-}\,{-}\,{-}\,{-}%
      \,{-}\,{-}\,{-}\,{-}\,{-}\,{-}\,{-}\,{-}\,{-}\,{-}\,{-}%
      \,{-}\,{-}\,{-}\,{-}\,{-}\,{-}\,{-}\,{-}\,{-}\,{-}\,{-}\,{-}\,{-}}\\
      \quad \Gamma\ent \ETUnrec_B\,b\,b'\,a :_l B[ a ]
    \end{array}
    \\
    \begin{array}[b]{l}
      \ETUnrec\mathtt{Beta}_0:\\
      \quad\{l:\NNO\}\\
      \quad\{B : \Ty{1}\}\\
      \quad\{b: \TM\}\\
      \quad\{b' : \Tm{2}\}\\
      \quad(S :\FsetA)\\
      \quad(\_: \Gamma \ent b :_l B [\ETUzero])\\
      \quad(\_ : \forall x\,y \fun x\freshfor y\freshfor S \fun
      \Gamma\cons {x :_0 \ETUNat}\cons {y :_l B[ x ]} \ent
      b'[x][y] :_l B[ \ETUsucc\,x ])\\
      \quad\text{{-}\,{-} helper hypothesis}\\
      \quad(\_ : \forall x \fun x\freshfor S \fun \Gamma\cons {x:_0
      \ETUNat} \ent B[x] :\ETUU_l)\\
      \quad{\fun}\;\text{{-}\,{-}\,{-}\,{-}\,{-}\,{-}\,{-}\,{-}\,{-}%
      \,{-}\,{-}\,{-}\,{-}\,{-}\,{-}\,{-}\,{-}\,{-}\,{-}\,{-}\,{-}%
      \,{-}\,{-}\,{-}\,{-}\,{-}\,{-}\,{-}\,{-}\,{-}\,{-}\,{-}\,{-}%
      \,{-}\,{-}\,{-}\,{-}\,{-}\,{-}\,{-}\,{-}\,{-}\,{-}\,{-}\,{-}%
      \,{-}\,{-}\,{-}\,{-}\,{-}\,{-}\,{-}\,{-}\,{-}}\\
      \quad \Gamma\ent \ETUnrec_B\,b\,b'\,\ETUzero = b :_l B[ \ETUzero ]
    \end{array}
    \\
    \begin{array}[b]{l}
      \ETUnrec\mathtt{Beta}_+:\\
      \quad\{l:\NNO\}\\
      \quad\{B : \Ty{1}\}\\
      \quad\{b\,a: \TM\}\\
      \quad\{b' : \Tm{2}\}\\
      \quad(S :\FsetA)\\
      \quad(\_: \Gamma \ent b :_l B [\ETUzero])\\
      \quad(\_ : \forall x\,y \fun x\freshfor y\freshfor S \fun
      \Gamma\cons {x :_0 \ETUNat}\cons {y :_l B[ x ]} \ent
      b'[x][y] :_l B[ \ETUsucc\,x ])\\
      \quad(\_ : \Gamma \ent a :_0 \ETUNat)\\
      \quad\text{{-}\,{-} helper hypothesis}\\
      \quad(\_ : \forall x \fun x\freshfor S \fun \Gamma\cons {x:_0
      \ETUNat} \ent B[x] :\ETUU_l)\\
      \quad{\fun}\;\text{{-}\,{-}\,{-}\,{-}\,{-}\,{-}\,{-}\,{-}\,{-}%
      \,{-}\,{-}\,{-}\,{-}\,{-}\,{-}\,{-}\,{-}\,{-}\,{-}\,{-}\,{-}%
      \,{-}\,{-}\,{-}\,{-}\,{-}\,{-}\,{-}\,{-}\,{-}\,{-}\,{-}\,{-}%
      \,{-}\,{-}\,{-}\,{-}\,{-}\,{-}\,{-}\,{-}\,{-}\,{-}\,{-}\,{-}%
      \,{-}\,{-}\,{-}\,{-}\,{-}\,{-}\,{-}\,{-}\,{-}}\\
      \quad \Gamma\ent 
      \ETUnrec_B\,b\,b'\,(\ETUsucc\,a) =
      b'[a][\ETUnrec_B\,b\,b'\,a] :_l B[ \ETUsucc\,a ]
    \end{array}
    \\
  \end{mathpar}
  \caption{Some of the rules for the eliminator for the natural number type}
  \label{fig:nrec}
\end{figure}

Figure~\ref{fig:rulee} gives the constructor functions for extensional
equality and the empty type. Note that some of those functions are to
do with congruence properties of the syntax constructions; it is
common to omit such rules in an informal textual description, but of
course they cannot be left out of a fully formal development. Note
also that some of the functions seem to have unnecessary arguments,
the ones marked with comments\footnote{Agda uses two or more
  consecutive hyphens {-}{-} to mark comments.} as ``helper
hypotheses''. They are indeed ultimately unnecessary since one can
prove that if $\Gamma\ent a:_lA$ is valid then so is
$\Gamma\ent A :\ETUU_l$ (see \agdalink{ETU.Admissible}). But along the
way to proving that and other useful properties of the type system it
is helpful to reason by induction on the structure of elements of the
judgement datatypes (rather than by induction on some notion of their
size) and these and other helper hypotheses aid this (cf.~the comments
by \citet[page~23:9]{AbelA:decctt}).

Figure~\ref{fig:nrec} gives some of the constructor functions
associated with $\ETUnrec$ terms (there are also constructor functions
for congruence properties, not shown in the figure), illustrating the
use of the concretion operation on term abstractions mentioned
above. The 7th and 9th arguments of ${\ent}\ETUnrec$ also illustrate
the way \emph{cofinite quantifications} such as
$\forall x \fun x\freshfor S \fun{}$ are used to deal with freshness
conditions in rule hypotheses. Elements of the type $\FsetA$ code
finite sets of variables (see \agdalink{WSLN.Atom}) and the datatype
$x\freshfor S$ contains proofs that variable $x:\Atom$ is not in the
set $S:\FsetA$; and similarly for $x\freshfor y \freshfor S$ (see
\agdalink{WSLN.Fresh}). Although initiated elsewhere, this use of
cofinite quantification in the context of language meta-theory is
convincingly advocated by \citet{PierceB:engfm}. It certainly makes it
easier to prove some of the properties of \ETU, such as weakening (see
\agdalink{ETU.Weakening}). Finitary versions of the typing rules,
eliminating cofinite quantification in favour of explicit freshness
hypotheses, are inter-derivable (and needed); see
\agdalink{ETU.ExistsFresh}.

The syntactic properties of \ETU{} make use of other forms of judgement
which are definable in terms of \eqref{eq:84}--\eqref{eq:86} (see
\agdalink{ETU.Rules}):
\begin{align}
  &{}\ent \Gamma = \Gamma'
  &&\text{$\Gamma$ and $\Gamma'$ are definitionally equal
     contexts}\label{eq:102} \\
  &\Gamma'\ent\sigma:\Gamma
  &&\text{in context $\Gamma'$, $\sigma$ is a well-typed substitution
     of type $\Gamma$}\label{eq:103}\\
  &\Gamma'\ent\sigma = \sigma':\Gamma
  &&\text{in context $\Gamma'$, $\sigma$ and $\sigma'$ are}\notag\\
  &&&\text{definitionally equal substitutions
     of type $\Gamma$}\label{eq:104}
\end{align}
Here $\sigma$ and $\sigma'$ are functions $\Atom\fun\TM$ from
variables to terms. One can show that provable judgements respect
definitional equality of contexts (see \agdalink{ETU.Admissible})
\begin{gather}
  ({}\ent \Gamma' = \Gamma) \fun (\Gamma\ent a :_l A) \fun (\Gamma'
  \ent a :_l A)\label{eq:105} \\
  ({}\ent \Gamma' = \Gamma) \fun (\Gamma\ent a = a' :_l A) \fun (\Gamma'
  \ent a = a' :_l A)\label{eq:106}\\
  \intertext{ and are stable under substitution (see
    \agdalink{ETU.Substitution})} 
  (\Gamma'\ent\sigma:\Gamma) \fun (\Gamma\ent a:_l A) \fun (\Gamma'\ent
  \sigma\ast a :_l \sigma\ast A)\label{eq:107}\\
  (\Gamma'\ent\sigma:\Gamma) \fun (\Gamma\ent a=a':_l A) \fun (\Gamma'\ent
  \sigma\ast a = \sigma\ast a' :_l \sigma\ast A)\label{eq:108}
\end{gather}
Here $\sigma\ast a$ denotes the term resulting from the application of
a substitution $\sigma$ to a term $a$ (defined in
\agdalink{WSLN.Sig.Substitution}). The proof of \eqref{eq:107} and
\eqref{eq:108} uses weakening 
\begin{gather}
  (\Gamma'\triangleright\Gamma) \fun (\Gamma\ent a:_lA) \fun
  (\Gamma'\ent a:_l A) \label{eq:142}\\
  (\Gamma'\triangleright\Gamma) \fun (\Gamma\ent a= a':_lA) \fun
  (\Gamma'\ent a = a':_l A) \label{eq:143}
\end{gather}
where the type $\Gamma'\triangleright\Gamma$ (defined in
\agdalink{ETU.Rules}) expresses that $\Gamma'$ is obtained by adding
hypotheses to the context $\Gamma$; see \agdalink{ETU.Weakening}.

\section{Semantics of \ETU}
\label{sec:semetu}

In this section we show how to give a semantics within \IRU{} to the
well-formed contexts, types and terms of \ETU{} that is sound for its
definitional equality, using the setoid universes from
Section~\ref{sec:setu}. In addition to the setoid universe $\C$ for
contexts described in Section~\ref{sec:setuc} we will need to define
displayed setoids indexed by universe levels $l:\NNO$:
\begin{align}
  &\Fam_l: \Setd[ \C ]
  &&\text{for interpreting \ETU{} types}\label{eq:109}\\
  &\Elem_l:\Setd[ \C\ltimes \Fam_l]
  &&\text{for interpreting \ETU{} terms}\label{eq:110}
\end{align}
(Recall setoid comprehension $\ltimes$ from Definition~\ref{def:com}.)
Before defining these displayed setoids, we sketch what the semantics
looks like in terms of them. It will take some effort to reach the
definition of the following functions (which can be found in
\agdalink{Semantics.Function}).
\begin{enumerate}
  
\item[(S1)] \textbf{\ETU{} contexts}: the semantics of a context
  $\Gamma$ is an element of the universe for contexts,
  $\dencx{\Gamma}_p : |\C|$, dependent upon a proof
  $p: (\Ok\,\Gamma)$. The semantics is \emph{proof-irrelevant} in the
  sense that if $p\,p': (\Ok\,\Gamma)$, then
  $\C\ni \dencx{\Gamma}_p \eq \dencx{\Gamma}_{p'}$.

\item[(S2)] \textbf{\ETU{} types}: the semantics of a type $A$ of
  universe level $l$ in context $\Gamma$ is an element
  $\denty{\Gamma\ent_l A}_{p,q} : \|\Fam_l\|\,\dencx{\Gamma}_p$
  dependent upon proofs $p:(\Ok\,\Gamma)$ and
  $q:(\Gamma\ent A :\ETUU_l)$. The semantics is
  \emph{proof-irrelevant} in the sense that if $p,p':(\Ok\,\Gamma)$
  and $q,q':(\Gamma\ent A :\ETUU_l)$, then
  $\Fam_l\ni \dencx{\Gamma}_p\mathrel{,} \denty{\Gamma\ent_l A}_{p,q}
  \heq \dencx{\Gamma}_{p'}\mathrel{,} \denty{\Gamma\ent_l A}_{p',q'}$.
  
\item[(S3)] \textbf{\ETU{} terms}: the semantics of a term $a$ of type
  $A$ of universe level $l$ in context $\Gamma$ is an element
  $\dentm{\Gamma\ent_l a}_{p,q,r} : \|\Elem_l\|(\dencx{\Gamma}_p
  \mathrel{,} \denty{\Gamma\ent_l A}_{p,q})$ dependent upon proofs
  $p:(\Ok\,\Gamma)$, $q:(\Gamma\ent A :\ETUU_l)$ and
  $r:(\Gamma\ent a:_l A)$.  The semantics is \emph{proof-irrelevant}
  in the sense that if $p\,p': (\Ok\,\Gamma)$,
  $q, q':(\Gamma\ent A :\ETUU_l)$ and $r, r':(\Gamma\ent a:_l A)$,
  then
  \[
    \Elem_l\ni
      (\dencx{\Gamma}_p\mathrel{,} \denty{\Gamma\ent_l
      A}_{p,q}) \mathrel{,} \dentm{\Gamma\ent_l a}_{p,q,r} 
      \heq (\dencx{\Gamma}_{p'}\mathrel{,} \denty{\Gamma\ent_l A}_{p',q'})
      \mathrel{,} \dentm{\Gamma\ent_l a}_{p',q',r'}
  \]

\item[(S4)] \textbf{\ETU{} definitional equality}: given
  $p:(\Ok\,\Gamma)$, $q:(\Gamma\ent A :\ETUU_l)$,
  $r : (\Gamma \ent a :_l A)$ and $r': (\Gamma \ent a' :_l A)$, then
  we have the following \emph{soundness} property\footnote{using the notation
    for the setoid fibres $B\fib x$ of a displayed setoid $B$ from
    Remark~\ref{rem:fibds}}
  \[
    (\Gamma\ent a = a':_l A) \fun
    \Elem_l
    \fib  (\dencx{\Gamma}_p\mathrel{,} \denty{\Gamma\ent_l
      A}_{p,q}) \ni \dentm{\Gamma\ent_l a}_{p,q,r} \eq
    \dentm{\Gamma\ent_l a'}_{p,q,r'}
  \]
  Similarly, given $q':(\Gamma\ent A' :\ETUU_l)$ 
  \[
    (\Gamma\ent A = A' :\ETUU_l) \fun
   \Fam_l \fib \dencx{\Gamma}_p \ni \denty{\Gamma\ent_l
      A}_{p,q} \eq \denty{\Gamma\ent_l A'}_{p,q'}
  \]
\end{enumerate}

\subsection{Families and their elements}
\label{sec:fame}

We now define $\Fam_l$ \eqref{eq:109} and $\Elem_l$
\eqref{eq:110}. Recall from Section~\ref{sec:extu} that \ETU{} types are
just particular \ETU{} terms and that the judgements for types and
their definitional equality are abbreviations for judgements about
terms~\eqref{eq:111}. Accordingly, for each context code $C:|\C|$ we
want $\|\Fam_l\|\,C$ to be definitionally equal in \IRU{} to
$\|\Elem_l\|(C , \univ_l)$ where $\univ_l : \|\Fam_{1{+} l}\|\,C$ is a
code (to be defined) for the level-$l$ universe;
cf.~\citet[section~2.3]{CoquandT:canndt}. To achieve this we
define both $\Fam_l$ and $\Elem_l$ in terms of a common notion, namely
the type $\Setd[ \E\fib C \Vdash F \ast \El_l ]$ of sections (in the
sense of Definition~\ref{def:dsetd}) of the re-indexing of
\eqref{eq:114} along a setoid morphism $F: |\E\fib C \fun \U_l|$.

\begin{defn}[$\Fam$]
  \label{def:fam}
  For each universe level $l:\NNO$ and context code $C:|\C|$ there is
  a setoid morphism $\univ_l : |\E\fib C \fun \U_{1{+}l}|$ which is
  constantly the element $\Univ : |\U_{1{+}l}|$ 
  \eqref{eq:48}. We define
  \begin{equation}
    \label{eq:117}
    \|\Fam_l\|\,C = \Setd[ \E\fib C \Vdash \univ_l \ast \El_{1{+}l} ]
  \end{equation}
  Thus each $T: \|\Fam_l\|\,C$ has an underlying function $\|T\|$
  \eqref{eq:16} which in view of \eqref{eq:52} has type
  $|\E\fib C| \fun |\U_l|$; and this function has a congruence
  property $\hcng\,T$ \eqref{eq:17} which in view of \eqref{eq:58} has
  type
  $\forall c\,c' \fun (\E\fib C \ni c \eq c') \fun (\U_l \ni \|T\|\,c
  \eq \|T\|\,c')$.  In other words, the elements of 
  $\|\Fam_l\|\,C$ are the setoid morphisms $\E\fib C \fun \U_l$ and in
  particular we have
  \begin{equation}
    \label{eq:31}
    \univ_l : \|\Fam_{1{+}l}\|\,C
  \end{equation}
  We get a displayed setoid $\Fam_l : \Setd[ \C ]$ by defining
  \begin{equation}
    \label{eq:28}
    (\Fam_l \ni C, T \heq C, T') =  \forall c\,c' \fun (\E\ni C , c
    \heq C', c') \fun (\U_l \ni \|T\|\,c \eq \|T'\|\,c')
  \end{equation}
  (see \agdadef{Semantics.CwF}{Fam}{Fam} for the definitions of the
  terms $\hrfl (\Fam_l)$, $\hsym (\Fam_l)$, $\htrs (\Fam_l)$,
  $\coe (\Fam_l)$ and $\coh (\Fam_l)$). Note that if
  $T: \|\Fam_l\|\,C$, then using the constructor $\Sigt$ \eqref{eq:71}
  we get a code for an extended context $C\ltimes_l T : |\C|$ where
  \begin{equation}
    \label{eq:81}
    C\ltimes_l T = \Sigt\,C\,l\,\|T\|\,(\hcng\,T)
  \end{equation}
\end{defn}

\begin{defn}[$\Elem$]
  \label{def:elem}
  For each universe level $l:\NNO$, context code $C:|\C|$ and family
  $T: \|\Fam_l\|\,C$, we noted above that $T$ is a setoid morphism
  from $\E\fib C \fun \U_l$; re-indexing $\El_l$ along this morphism
  we define
  \begin{equation}
    \label{eq:29}
    \|\Elem_l\|(C,T) = \Setd[ \E\fib C \Vdash T \ast \El_l ]
  \end{equation}
  Thus each $t: \|\Elem_l\|(C,T)$ has an underlying function $\|t\|$
  \eqref{eq:16}  of type $(c: |\E\fib C|)\fun \|\El_l\|(\| T \| \, c)$
  and a congruence property $\hcng\,t$ \eqref{eq:17} of type
  $\forall c\,c'\fun (\E\fib C \ni c \eq c') \fun (\El_l\ni \|T\|\,c ,
  \|t\|\,c \heq \|T\|\,c' , \|t\|\, c')$. We get a displayed setoid
  $\Elem_l:\Setd[ \C\ltimes \Fam_l]$ by defining
  \begin{multline}
    \label{eq:30}
    (\|\Elem_l\|\ni (C,T)\,t \heq (C',T'),t')
    = {}\\
    \forall c\,c'\fun (\E\ni C, c \heq C' , c') \fun
    (\El_l\ni \|T\|\,c , \|t\|\,c \heq \|T'\|\,c' , \|t'\|\, c')
  \end{multline}
  (see \agdadef{Semantics.CwF}{Elem}{Elem} for the definitions of the
  terms $\hrfl (\Elem_l)$, $\hsym (\Elem_l)$, $\htrs (\Elem_l)$,
  $\coe (\Elem_l)$ and $\coh (\Elem_l)$). Combining \eqref{eq:117} and
  \eqref{eq:29} we have the desired definitional equality between
  $\|\Fam_l\|\,C$ and $\|\Elem_{1{+}l}\|(C,\univ_l)$.
\end{defn}

The displayed setoids $\Fam_l$ \eqref{eq:109} and $\Elem_l$
\eqref{eq:110} inherit the structure that has been built into the
definition of the setoid universes from Section~\ref{sec:setu}, which
we will use in the next section to give a semantics for \ETU.  Here we
will just describe the structure that models extensional equality types
and the empty type; see \agdalink{Semantics.CwF} for the other type formers.

\begin{defn} 
  \label{def:eqt}
  (See \agdadef{Semantics.CwF}{EqualityType}{EqualityType}.)
  Given $l:\NNO$, $(C,T):|\C\ltimes \Fam_l|$ and $t,t':
  \|\Elem_l\|(C,T)$, there is a family
  \begin{equation}
    \label{eq:83}
    \Eqt_l\,T\,t\,t' : \|\Fam_l\|\,C
  \end{equation}
  where, for each $c:|\E\fib C|$, we use the constructor
  \eqref{eq:22} in case $l=0$ and \eqref{eq:50} in case $l$ is a successor:
  \[
    \begin{array}{l}
      \|\Eqt_0\,T\,t\,t'\|\,c
      = \Eq_0(\|T\|\,c)(\|t\|\,c)(\|t'\|\,c)\\
      \|\Eqt_{1{+}\_}\,T\,t\,t'\|\,c
      = \Eq_+(\|T\|\,c)(\|t\|\,c)(\|t'\|\,c)\\ 
    \end{array}
  \]
  There is an element for reflexivity
  \begin{equation}
    \label{eq:112}
    \rflt_l\,T\,t :\|\Elem_l\|(C, \Eqt_l\,T\,t\,t)
  \end{equation}
  with $\|\rflt_l\,T\,t \|\,c$ given by $\hrfl(\|T\|\,c)(\|t\|\,c)$
  (using \eqref{eq:36} and \eqref{eq:75}). Furthermore the equality
  family satisfies reflection
  \begin{equation}
    \label{eq:118}
    \reflect_l : \|\Elem_l\|(C, \Eqt_l\,T\,t\,t') \fun (\Elem_l\fib(C,T)
    \ni t \eq t')
  \end{equation}
  Indeed, because of the way $\Eq_0$ and $\Eq_+$ are defined, both
  $\reflect_0$ and $\reflect_{1{+}\_}$ are the identity
  function. Finally, the equality family trivially satisfies uniqueness
  of identity proofs
  \begin{equation}
    \label{eq:119}
    \uip_l :(e\,e': \|\Elem_l\|(C, \Eqt_l\,T\,t\,t')) \fun
    (\Elem_l\fib (C, \Eqt_l\,T\,t\,t') \ni e \eq e')
  \end{equation}
  because, both in case $l=0$ or it is a successor,
  $(\Elem_l\fib (C, \Eqt_l\,T\,t\,t') \ni e \eq e')$ is the one-element
  type $\top$ (by virtue of \eqref{eq:43} and \eqref{eq:60}).
\end{defn}

\begin{defn}
  \label{def:emt}
  (See \agdadef{Semantics.CwF}{EmptyType}{EmptyType}.)
  Given $C:|\C|$ there is a family
  \begin{equation}
    \label{eq:140}
    \Empt : \|\Fam_0\|\,C
  \end{equation}
  given by the constant morphism $\E\fib C \morphism \U_0$ with value
  $\Emp:|\U_0|$ \eqref{eq:9}. Note that by virtue of \eqref{eq:34} and
  \eqref{eq:29}, for any $t: \|\Elem_0\|(C,\Empt)$, $\|t\|$ is a
  function of type $|\E\fib C| \fun \emptyset$ and so if $c: |\E\fib
  C|$ then $\|t\|\,c :\emptyset$ can be mapped to any type by the
  eliminator in \IRU{} for the empty type $\emptyset$. Consequently for
  any such $t$, given $l:\NNO$ and $T:\|\Fam_l\|\,C$ we get an
  element
  \begin{equation}
    \label{eq:141}
    \empt_l\,T\, t : \|\Elem_l\|(C,T)
  \end{equation}
  with $\|\empt_l\,T\, t\|\,c$ eliminating $\|t\|\,c :\emptyset$  to
  the type $\El_l (\|T\|\,c)$.
\end{defn}

\subsection{The semantic relations}
\label{sec:semr}

To prove (S1)--(S4) one could try to proceed by induction on the
structure of the proofs of the \ETU{} judgements. Since the three basic
forms of judgement \eqref{eq:84}--\eqref{eq:86} are mutually
dependent, one would have to construct the semantic functions in
(S1)--(S3) simultaneously, at the same time as building up proofs of
their proof-irrelevance property and the soundness property (S4). All
this ``would result in a more complicated proof''\footnote{I am
  quoting \citet[page~43]{HofmannM:synsdt} writing about a similar
  strategy for the semantics in the models considered in that survey
  article.}  than the one we are going to give; indeed it is
sufficiently complicated that after some attempts the author has no
concrete evidence that it actually exists.\footnote{To be convincing
  one wants a machine-checked formalization.} Instead, and like
\citet{HofmannM:synsdt}, we adopt the original strategy
of \citet{StreicherT:semtt} of first giving \emph{partial} semantic
functions on syntax that do not depend upon proofs of the judgements
and then showing their totality where proofs of the relevant
judgements do exist. Since we work entirely within \IRU{} there are
limited means for expressing partial functions: we have to define
type-valued input-output relations and then prove that they are not
only suitably total, but also single-valued (in the world of setoids).

The semantic relations for contexts and for terms that we use take the
form of datatypes
\begin{gather} 
  (\dencxr{\Gamma}C) : \Set\label{eq:63}\qquad\\
  (\dentmr{\Gamma\ent_l a}((C,T),t)) : \Set\label{eq:64} 
\end{gather}
where $\Gamma$ is an \ETU{} context, $a$ an \ETU{} term, $l:\NNO$ and
$((C,T),t) : |\C \ltimes \Fam_l \ltimes \Elem_l|$.  The semantic
relation for types is treated as a special case of that for terms
(making use of the identification of $\|\Fam_l\|\,C$ with
$\|\Elem_{1{+}l}\|(C,\univ_l)$):
\begin{equation}
  \label{eq:65}
  \dentyr{\Gamma\ent_l A}(C,T) = \dentmr{\Gamma\ent_l A}((C ,
  \univ_l), T)
\end{equation}
The datatype \eqref{eq:63} has two constructors, corresponding to the
two ways of building contexts, \eqref{eq:99} and \eqref{eq:100}; their
types are:
\begin{gather}
  \dencxr{\ETUECx}\Unit\label{eq:66}\\
  \dentyr{\Gamma\ent_l A}(C,T) \fun
  (x\freshfor \Gamma) \fun \dencxr{\Gamma\cons{x:_lA}}{C \ltimes_l
    T} \label{eq:80} 
\end{gather}
(the first uses the constructor $\Unit$ \eqref{eq:70}; the second uses
a freshness relation $x\freshfor \Gamma$ (expressing that $x$ does not
occur in $\Gamma$) and the comprehension construct $\ltimes_l$
mentioned at the end of Definition~\ref{def:fam}).

\begin{figure}
  \renewcommand{\arraystretch}{.8}
  \begin{mathpar}
    \begin{array}[b]{l}
      \den{\ETUEq} :\\
      \quad\{l:\NNO\}\\
      \quad\{A\,a\,a' : \TM\}\\
      \quad\{C:|\C|\}\\
      \quad\{T:\|\Fam_l\|\,C\}\\
      \quad\{t\,t': \|\Elem\_l\|(C,T)\}\\
      \quad(\_: \dentyr{\Gamma\ent_l A}(C,T))\\
      \quad(\_:\dentmr{\Gamma\ent_l a}((C,T),t))\\
      \quad(\_:\dentmr{\Gamma\ent_l a'}((C,T),t'))\\
      \quad{\fun}\;\text{{-}\,{-}\,{-}\,{-}\,{-}\,{-}\,{-}%
      \,{-}\,{-}\,{-}\,{-}\,{-}\,{-}\,{-}\,{-}\,{-}\,{-}%
      \,{-}\,{-}\,{-}\,{-}\,{-}}\\
      \quad\dentyr{\Gamma\ent_l \ETUEq_A\,a\,a'}(C, \Eqt_l\,T\,t\,t')
    \end{array}
    \and
    \begin{array}[b]{l}
      \den{\ETUrefl} :\\
      \quad\{l:\NNO\}\\
      \quad\{A\,a: \TM\}\\
      \quad\{C:|\C|\}\\
      \quad\{T:\|\Fam_l\|\,C\}\\
      \quad\{t : \|\Elem\_l\|(C,T)\}\\
      \quad(\_ : \dentyr{\Gamma\ent_l A}(C,T))\\
      \quad(\_ :\dentmr{\Gamma\ent_l a}((C,T),t))\\
      \quad{\fun}\;\text{{-}\,{-}\,{-}\,{-}\,{-}\,{-}%
      \,{-}\,{-}\,{-}\,{-}\,{-}\,{-}\,{-}\,{-}\,{-}%
      \,{-}\,{-}\,{-}\,{-}\,{-}\,{-}\,{-}\,{-}\,{-}%
      \,{-}\,{-}\,{-}\,{-}\,{-}}\\
      \quad\dentmr{\Gamma\ent_l \ETUrefl_A\,a}((C, \Eqt_l\,T\,t\,t), \rflt_l\,T\,t)
    \end{array}
    \\
  \end{mathpar}

  \caption{Semantic relation constructors for equality types
    and reflexivity terms}
  \label{fig:consre}
\end{figure}

\begin{figure}
  \renewcommand{\arraystretch}{.8}
  \begin{mathpar}
    \begin{array}[b]{l}
      \den{\ETUEmp} :\\
      \quad\{C:|\C|\}\\
      \quad(\_ : \dencxr{\Gamma} C)\\
      \quad{\fun}\;\text{{-}\,{-}\,{-}\,{-}\,{-}\,{-}%
      \,{-}\,{-}\,{-}\,{-}\,{-}\,{-}\,{-}\,{-}\,{-}\,{-}}\\
      \quad\dentyr{\Gamma\ent_0 \ETUEmp}(C, \Empt)
    \end{array}
    \and
    \begin{array}[b]{l}
      \den{\ETUemp} :\\
      \quad\{l:\NNO\}\\
      \quad\{a\,B: \TM\}\\
      \quad\{C:|\C|\}\\
      \quad\{T:\|\Fam_l\|\,C\}\\
      \quad\{t : \|\Elem_0\|(C, \Empt)\}\\
      \quad(\_ : \dentyr{\Gamma \ent_l B}(C , T))\\
      \quad(\_ : \dentmr{\Gamma\ent_0 a}((C,\Empt),t))\\
      \quad{\fun}\;\text{{-}\,{-}\,{-}\,{-}\,{-}\,{-}\,{-}%
      \,{-}\,{-}\,{-}\,{-}\,{-}\,{-}\,{-}\,{-}\,{-}\,{-}%
      \,{-}\,{-}\,{-}\,{-}\,{-}\,{-}\,{-}\,{-}}\\
      \quad\dentmr{\Gamma\ent_l \ETUemp_B\,a}((C, T),\empt_l\,T\,t)
    \end{array}
    \\
  \end{mathpar}
  \caption{Semantic relation constructors for the empty type
    and its eliminator}
  \label{fig:semrce}
\end{figure}

The datatype \eqref{eq:64} has a constructor for each of the ways
\eqref{eq:87}--\eqref{eq:101} of building \ETU{} terms.
Note that $\Pi$-types in \ETU{} obey
Agda-style universe levels: if $\Gamma\ent A:\ETUU_l$ and
$\Gamma\cons{x:_l A} \ent B[x ] : \ETUU_{l'}$, then
$\Gamma\ent \ETUPi_{l,l'}\,A\,B : \ETUU_{\max\,l\,l'}$. The semantics
of $\ETUPi_{l,l'}\,A\,B$ uses the lifting operation \eqref{eq:49} that
is part of the definition of the setoid universes (see
\agdalink{Setoid.Lift}) to move the semantics of $A$ and $B$ to the
common level $\max\,l\,l'$ before applying the $\Pit_+$ constructor
\eqref{eq:51}. We refer the reader to \agdalink{Semantics.Relation}
for the full details. Figure~\ref{fig:consre} gives the constructors
for equality types and reflexivity terms using
Definition~\ref{def:eqt}; and Figure~\ref{fig:semrce} gives the
constructors for the empty type and its eliminator terms using
Definition~\ref{def:emt}. 

We need \eqref{eq:64} to respect the equivalence of the setoid
$\C \ltimes \Fam_l \ltimes \Elem_l$
\begin{equation}
  \label{eq:82}
  \dentmr{\Gamma\ent_l a}x \fun (\C \ltimes \Fam_l \ltimes \Elem_l
  \ni x \eq x') \fun \dentmr{\Gamma\ent_l a}x'
\end{equation}
This property could be guaranteed by building a version of it into the
definition of each constructor. However, we found it easier to use
simpler definitions of the constructors that do not necessarily
respect setoid equivalence and include property \eqref{eq:82} as a
further constructor of the datatype. (A similar respectfulness property
for the datatype \eqref{eq:63} holds automatically without having to
give an explicit constructor for it.)

\begin{thm}
  \label{thm:sv}
  The semantic relations are single-valued:
  \begin{gather}
    \dencxr{\Gamma}x \fun \dencxr{\Gamma}x' \fun (\C \ni x \eq
    x') \label{eq:120}\\
    \dentyr{\Gamma\ent_l A}x \fun \dentyr{\Gamma\ent_l A}x' \fun
    (\C\ltimes\Fam_l \ni x \eq x')\label{eq:136}\\
    \dentmr{\Gamma\ent_l a}x \fun \dentmr{\Gamma\ent_l a}x' \fun
    (\C\ltimes\Fam_l\ltimes\Elem_l \ni x \eq x')\label{eq:121}
  \end{gather}
\end{thm}
\begin{proof}
  See \agdalink{Semantics.SingleValued}.  \eqref{eq:136} follows from
  \eqref{eq:120} and \eqref{eq:121}; and those two can be proved
  simultaneously by structural induction for the datatypes
  $\dencxr{\Gamma}x$ and $\dentmr{\Gamma\ent_l a}x$.  The proof 
  makes use of the fact that \ETU{} terms are decorated with explicit
  universe-level and type information.
\end{proof}

\subsection{Semantics of substitution}
\label{sec:sems}

Recall from the end of Section~\ref{sec:extu} that we take
term-for-variable substitutions in \ETU{} just to be functions
$\Atom\fun\TM$ and define judgements \eqref{eq:103} and
\eqref{eq:104} about them in terms of the basic judgements. To model
those judgements we use a notion of morphism between context codes in
$\C$. There is a displayed setoid $\Hom : \Setd[ \C \otimes \C ]$
given as follows (see \agdadef{Semantics.CwF}{Hom}{Hom}):
\begin{equation}
  \label{eq:124}
  \begin{aligned}
    &\|\Hom\|\,(C , D) = | \E\fib C \morphism \E\fib D |\\
    &(\Hom \ni (C, D), f\heq (C' , D'), f') = {}\\
    &\qquad\forall c\,c' \fun (\E\ni C,c\heq C', c') \fun (\E\ni D,
    |f|\,c \heq D', |f'|\,c') 
  \end{aligned}
\end{equation}
In Definition~\ref{def:rei} we noted that sections of displayed
setoids can be re-indexed along setoid morphisms. Since families
\eqref{eq:117} and their elements \eqref{eq:29} are defined in terms
of sections, we get induced operations for re-indexing a family
$T:\|\Fam_l\|\,C$ or an element $t:\|\Elem_l\|(C,T)$ along a morphism
$f:\|\Hom\|(D,C)$ (see
\agdadef{Semantics.CwF}{ReIndexFam}{ReIndexFam}):
\begin{equation}
  \label{eq:128}
  \begin{aligned}
    &f \ast T : \|\Fam_l\|\,D\\
    &f\ast_T t : \|\Elem_l\|(D, f\ast T)
  \end{aligned}
\end{equation}
The semantic relation for substitutions takes the form of a datatype in
$\Set$ 
\begin{equation}
  \label{eq:125}
  \densbr{\Delta \ent \sigma : \Gamma}((D,C),f)
\end{equation}
where $\Delta$ and $\Gamma$ are \ETU{} contexts, $\sigma:\Atom\fun\TM$
and $((D,C),f) : |(\C\otimes \C)\ltimes \Hom|$. It has constructors
for the two ways of constructing a substitution judgement, plus a
constructor ensuring that \eqref{eq:125} respects the equivalence of
the setoid $(\C\otimes \C)\ltimes \Hom$; see
\agdalink{Semantics.Substitution}.

\begin{lem}
  \label{lem:sb}
  The operation of applying a substitution $\sigma:\Atom\fun\TM$ to an
  \ETU{} term $a$ to obtain another term $\sigma\ast a$ has a semantics
  given by the re-indexing operations \eqref{eq:128}, in the sense
  that there is a function
  \[
    \densbr{\Delta \ent \sigma : \Gamma}((D,C),f) \;\fun\;
    \dentmr{\Gamma\ent_l a}((C,T),t) \;\fun\;
    \dentmr{\Delta\ent_l \sigma \ast a}((D,f\ast T),f\ast_T t)
  \]
\end{lem}
\begin{proof}
  See \agdalink{Semantics.Substitution}. (The proof uses a semantics
  for the operation of weakening the context of a judgement, defined
  in \agdalink{Semantics.Weakening}.)
\end{proof}

\subsection{Totality and soundness}
\label{sec:tots}

The totality of the semantic relations for \ETU{} contexts and terms can
be expressed as the existence of the following functions:
\begin{gather}
  \totCx : (\Ok\,\Gamma)\fun
  \textstyle\sum[ C\in |\C|]\,\left(\dencxr{\Gamma}C\right)
  \label{eq:130}\\
  \totTy : (\Gamma\ent A: \ETUU_l) \fun
    \textstyle\sum[ (C,T)\in
    |\C\ltimes\Fam_l|]\,\left(\dentyr{\Gamma\ent_lA}(C,T)\right)
  \label{eq:131}\\ 
  \totTm : (\Gamma\ent a:_l A) \fun
  \begin{array}[t]{@{}l}
    \textstyle\sum[ ((C,T),t) \in |\C\ltimes\Fam_l\ltimes\Elem_l|]\\
    \left(\dentyr{\Gamma\ent_l A}(C,T)\right) \;\times\;
    \left(\dentmr{\Gamma\ent_l a}((C,T),t)\right)
  \end{array}
  \label{eq:132}
\end{gather}
The soundness of the semantic relations for \ETU{} definitional
equality can be expressed as the existence of the following functions:
\begin{gather}
  \soundTy : (\Gamma\ent A = A': \ETUU_l) \fun
  \begin{array}[t]{@{}l}
    \textstyle\sum[ (C,T) \in |\C\ltimes\Fam_l|]\\
    \left(\dentyr{\Gamma\ent_lA}(C,T)\right) \;\times\;
    \left(\dentyr{\Gamma\ent_lA'}(C,T)\right)
  \end{array}\label{eq:133}\\
  \soundTm : (\Gamma\ent a = a' :_l A) \fun
  \begin{array}[t]{@{}l}
    \textstyle\sum[ ((C,T),t) \in |\C\ltimes\Fam_l\ltimes\Elem_l|]\\
    \left(\dentyr{\Gamma\ent_l A}(C,T)\right) \;\times\;\\
    \left(\dentmr{\Gamma\ent_l a}((C,T),t) \right) \;\times\;\\
    \left(\dentmr{\Gamma\ent_l a'}((C,T),t) \right)
  \end{array}\label{eq:134}
\end{gather}

\begin{thm}
  \label{thm:tots}
  The terms \eqref{eq:130}--\eqref{eq:134} exist  in \IRU. 
\end{thm}
\begin{proof}
  The terms $\totTy$ and $\soundTy$ can be constructed from the other
  three terms (reflecting the fact that typing judgements in \ETU{}
  are defined in terms of the basic judgements for contexts and
  terms). The terms $\totCx$, $\totTm$ and $\soundTm$ have to be
  defined simultaneously, by induction on the proofs of \ETU{}
  judgements. The details of the proof are complicated and can be
  found in \agdalink{Semantics.Total}. They use the single-valuedness
  property (Theorem~\ref{thm:sv}) and substitution property
  (Lemma~\ref{lem:sb}), together with the following conditional
  versions which are definable from
  \eqref{eq:130}--\eqref{eq:134}:
  \begin{gather}
    \totTy' : (\Gamma\ent A: \ETUU_l) \fun \left(\dencxr{\Gamma}C\right)
    \fun{}\notag\\ 
    \qquad\textstyle
    \sum[T \in \|\Fam_l\|\,C]\,\left(\dentyr{\Gamma\ent_lA}(C,T)\right)
    \label{eq:122}\\
    \totTm' : (\Gamma\ent a:_l A) \fun
    \left(\dentyr{\Gamma\ent_l A}(C,T)\right)
    \fun{}\notag\\
    \qquad\textstyle\sum[ t \in \|\Elem_l\|(C,T)]\,
    \left(\dentmr{\Gamma\ent_l a}((C,T),t)\right) 
    \label{eq:123}\\
    \soundTy' (\Gamma\ent A = A': \ETUU_l) \fun
    \left(\dencxr{\Gamma}C\right)
    \fun{}\notag\\ 
    \qquad\textstyle\sum[T \in
    \|\Fam_l\|\,C]\,\left(\dentyr{\Gamma\ent_lA}(C,T)\right)  \;\times\;
    \left(\dentyr{\Gamma\ent_lA'}(C,T)\right)
    \label{eq:129}\\
    \soundTm' : (\Gamma\ent a = a':_l A) \fun
    \left(\dentyr{\Gamma\ent_l A}(C,T)\right) 
    \fun{}\notag\\
    \qquad\textstyle\sum[ t \in \|\Elem_l\|(C,T)]\,
    \left(\dentmr{\Gamma\ent_l a}((C,T),t)\right) \;\times\;
    \left(\dentmr{\Gamma\ent_l
      a'}((C,T),t)\right)
    \label{eq:135}
  \end{gather}
\end{proof}

\subsection{The semantic functions}
\label{sec:semf}

We can now construct the functions described in (S1)--(S4) and prove
their irrelevancy and soundness properties (see
\agdalink{Semantics.Function}).

For (S1) we can take
\begin{equation}
  \label{eq:137}
  \dencx{\Gamma}_p = \pi_1 (\totCx\,p)
\end{equation}

Thus $\pi_2(\totCx\,p) : \dencxr{\Gamma}{\dencx{\Gamma}_p}$; so if we
also have $p':(\Ok\,\Gamma)$, then we can apply Theorem~\ref{thm:sv}
to conclude that $\C\ni \den{\Gamma}_p \eq \den{\Gamma}_{p'}$.

For (S2) we can take
\begin{equation}
  \label{eq:138}
  \denty{\Gamma\ent_l A}_{p,q} = 
  \pi_1(\totTy'\,q\,(\pi_2(\totCx\,p)))
\end{equation}
Thus $\pi_2(\totTy'\,q\,(\pi_2(\totCx\,p)))$ is a term of type
$\dentyr{\Gamma\ent_l A}(\dencx{\Gamma}_p, \denty{\Gamma\ent_l
  A}_{p,q})$; so if we also have $p':(\Ok\,\Gamma)$ and
$q':(\Gamma\ent A: \ETUU_l)$, then we can apply Theorem~\ref{thm:sv}
to conclude that
$\Fam_l\ni \dencx{\Gamma}_p\mathrel{,} \denty{\Gamma\ent_l A}_{p,q}
\heq \dencx{\Gamma}_{p'}\mathrel{,} \denty{\Gamma\ent_l A}_{p',q'}$.

For (S3) we can take
\begin{equation}
  \label{eq:139}
  \dentm{\Gamma\ent_l a}_{p,q,r} = 
  \pi_1(\totTm'\,r\,(\pi_2(\totTy'\,q\,(\pi_2(\totCx\,p)))))
\end{equation}
Thus
\begin{multline*}
  \pi_2(\totTm'\,r\,(\pi_2(\totTy'\,q\,(\pi_2(\totCx\,p))))) :\\
  \dentmr{\Gamma\ent_l a}((\dencx{\Gamma}_p, \denty{\Gamma\ent_l
    A}_{p,q}), \dentm{\Gamma\ent_l a}_{p,q,r})
\end{multline*}
So if we also have
 $p':(\Ok\,\Gamma)$,
$q':(\Gamma\ent A: \ETUU_l)$ and $r':(\Gamma\ent a :_l A)$, then we
can apply Theorem~\ref{thm:sv}
to conclude that
\[
  \Elem_l\ni
  (\dencx{\Gamma}_p\mathrel{,} \denty{\Gamma\ent_l
\qed    A}_{p,q}) \mathrel{,} \dentm{\Gamma\ent_l a}_{p,q,r} 
  \heq (\dencx{\Gamma}_{p'}\mathrel{,} \denty{\Gamma\ent_l A}_{p',q'})
  \mathrel{,} \dentm{\Gamma\ent_l a}_{p',q',r'}
\]

For (S4), suppose we have  $p:(\Ok\,\Gamma)$, $q:(\Gamma\ent A :\ETUU_l)$,
$r : (\Gamma \ent a :_l A)$, $r': (\Gamma \ent a' :_l A)$ and
$s:(\Gamma\ent a = a' :_l A)$. Let
\[
  \begin{aligned}[t]
    (C,p_0) &= \totCx\,p\\
    (T, q_0) &= \totTy'\,q\,p_0\\
    (t , r_0) &= \totTm'\,r\,q_0\\
    (t' , r_0') &= \totTm'\,r'\,q_0\\
    (t'' , r_1 , r_1') &= \soundTm'\,s\,q_0
  \end{aligned}
  \text{hence}\quad
  \begin{aligned}[t]
    C &= \dencx{\Gamma}_p\\
    T &= \denty{\Gamma\ent_l A}_{p,q}\\
    t &= \dentm{\Gamma\ent_l a}_{p,q,r}\\
    t' &= \dentm{\Gamma\ent_l a'}_{p,q,r'}
  \end{aligned}
  \text{and}\quad
  \begin{aligned}[t]
    r_0 &: \dentmr{\Gamma\ent_l a}((C,T),t)\\
    r_0' &: \dentmr{\Gamma\ent_l a'}((C,T),t')\\
    r_1 &: \dentmr{\Gamma\ent_l a}((C,T),t'')\\
    r_1' &: \dentmr{\Gamma\ent_l a'}((C,T),t'')
  \end{aligned}
\]
So we can apply Theorem~\ref{thm:sv} first to $r_0$ and $r_1$ to
conclude that $\Elem_l\fib(C,T)\ni t \eq t''$; and then apply it to
$r_1'$ and $r_0'$ to conclude that $\Elem_l\fib(C,T)\ni t'' \eq
t'$. So by transitivity $\Elem_l\fib(C,T)\ni t \eq t'$, that is,
$\Elem_l \fib (\dencx{\Gamma}_p\mathrel{,} \denty{\Gamma\ent_l
  A}_{p,q}) \ni \dentm{\Gamma\ent_l a}_{p,q,r} \eq \dentm{\Gamma\ent_l
  a'}_{p,q,r'}$, as required for the first part of (S4). The proof of
the second part is similar. \qed

\subsection{Consistency}
\label{sec:consis}

Were there to exist an \ETU{} term $a:\TM$ satisfying
$\ETUECx\ent a:_0 \ETUEmp$, then the rules in Figure~\ref{fig:rulee}
together with weakening \eqref{eq:142} imply that every well-formed
\ETU{} type in context is inhabited; and hence by equality reflection,
that all terms of any type in context are definitionally equal. Thus
the \IRU{} type $\neg(\sum[ a \in \TM] (\ETUECx\ent a:_0 \ETUEmp))$
expresses the consistency of \ETU. Recall that provability of
judgements in extensional type theory is in general
undecidable~\citep[Theorem~3.2.1]{HofmannM:extcit}. So consistency
cannot be proved via some effective notion of normalization for which
the empty type has no normal forms. Instead one can proceed via a
semantics in a model for which the empty type is suitably empty. We
now have such a model in \IRU{} and so we get:

\begin{thm}
  \label{thm:consis}
  There is an \IRU{} term $\mathtt{consis_{\mathsf{ETU}}}$ of type
  $\neg(\sum[ a \in \TM] (\ETUECx\ent a:_0 \ETUEmp))$.
\end{thm}
\begin{proof} (See \agdalink{Semantics.Consistency}.)
  By definition the type $\Ok\,\ETUECx$ contains an element in
  constructor form, call it $p$; and then from Figure~\ref{fig:rulee}
  we get $q : (\ETUECx\ent \ETUEmp : \ETUU_0)$ given by
  $q = {\ent}\ETUEmp\, p$.  So if $a: \TM$ and
  $r : (\ETUECx\ent a:_0 \ETUEmp)$, then from (S3) we get
  $\dentm{\ETUECx\ent_0 a}_{p,q,r}$ in $\|\Elem_0\|(\dencx{\ETUECx}_p
  \mathbin{,} \denty{\ETUECx\ent_0 \ETUEmp}_{p,q})$. Its underlying
  function $\|\dentm{\ETUECx\ent_0 a}_{p,q,r}\|$ has type
  $\top\fun\emptyset$ because of the way that $\dencx{\ETUECx}_p$ and 
  $\denty{\ETUECx\ent_0 \ETUEmp}_{p,q}$ are defined. Therefore we can
  define $\mathtt{consis_{\mathsf{ETU}}}$ to be the function
  $(\sum[ a \in \TM] (\ETUECx\ent a:_0 \ETUEmp)) \fun \emptyset$ that
  maps $(a,r)$ to $\|\dentm{\ETUECx\ent_0 a}_{p,q,r}\|\,\tt$.
\end{proof}

\section{Conclusion}
\label{sec:con}

Peering over the shoulders of Altenkirch, Hofmann, Palmgren and
others, we have shown that a certain notion of displayed setoid for
type-valued equivalence relations (Definition~\ref{def:dsetd}) works
well in intensional type theory, to the extent that it is possible to
use it to give an intensional semantics for extensional type theory
with universes (Section~\ref{sec:semetu}). Although the semantics
generally follows the pattern pioneered by \citet{StreicherT:semtt}
intensionality complicates matters and there are novel aspects.
Although the use of setoid universes of codes for types is not new,
the use of a setoid universe of codes for contexts
(section~\ref{sec:setuc}) does seem new; and we had to be careful to
ensure that families were the same thing as elements of the code for a
universe (section~\ref{sec:fame}).  As a consequence of the setoid
semantics we deduced the consistency of extensional type theory as a
theorem within the intensional type theory, using safe Agda to fully
formalize it.

Much has been written about the formalization of the syntax and
semantics of dependent type theories.  Traditional formalizations of
syntax, such as we used here, take an ``extrinsic'' approach, with
inductive datatypes of raw expressions and inductively defined
judgements about which expressions are well-formed; only a very simple
intensional type theory, like \IRU, is needed to express such a
formalization. By contrast, more recent ``intrinsic'' accounts need
only define well-formed expressions, as the elements of a initial
algebra for a suitable (second order) generalized algebraic
theory~\citep{CartmellJ:genatc,UemuraT:absctt}.  Category theory can be
applied to get properties of those initial algebras and thus one gets
elegant proofs of meta-mathematical properties of the intrinsic
version of dependent type
theories~\citep{CoquandT:canndt,SterlingJ:firsst}. To recover results
about the extrinsic version, one has to prove ``initiality
conjectures'' about the correspondence of the traditional syntax
modulo definitional equality with elements of the initial algebra.

Here we have proved results entirely within \IRU, partly from a
foundational motivation, and partly from a practical one since
intensional type theory is well supported and well used in
``engineering'' formal meta-theory'' \citep{PierceB:engfm}. From this
point of view there are two problems with the more abstract, intrinsic
approach sketched in the previous paragraph: it involves constructing
initial algebras and that involves taking quotients (or better, using
quotient inductive types~\citep{AltenkirchT:typtt}); and solving
initiality conjectures involves using extensionality (think how the
uniqueness part of initiality is going to be expressed). Although
intensional type theory has neither quotient types nor extensionality
built in, previous work leads us to expect that their effect can be
regained via the use of some kind of setoid. But what kind exactly?
--- realizing that expectation is not automatic. For example, can one
use setoids to get a fully formalized intensional version of the
category-theoretic glueing treatment of
normalization-by-evaluation~\citep{FioreMP:semanbe-jv,CoquandT:canndt}?
The answer is ``yes'' for simply typed systems like G\"odel's System
T, see \citet[Example~6.3]{PittsAM:welsln}; but that just needs
setoids rather than displayed setoids. The material presented in this
paper is a step towards being able to answer such questions in the
dependently typed case.

\paragraph*{Acknowledgements} I am grateful to Lo{\"i}c Pujet for
pointing me to his work \citep{PujetL:indus,PujetL:revhsm}; and to him
and David Berry for discussions.


\end{document}